\documentclass[12pt]{article}
\usepackage[latin1]{inputenc}
\usepackage{lmodern}

\usepackage{amssymb, amsmath, amsthm}
\usepackage[a4paper,top=25mm,bottom=25mm,left=25mm,right=25mm]{geometry}
\usepackage{etex}

\usepackage{authblk} 
\usepackage{pifont}
\usepackage{graphicx}
\usepackage[usenames,dvipsnames,svgnames,table]{xcolor}
\usepackage[figuresright]{rotating}
\usepackage{xtab} 
\usepackage{longtable} 
\usepackage{footnote}
\usepackage[stable]{footmisc}
\usepackage{chngpage} 
\usepackage{pdflscape} 

\usepackage{pgfplots}
\usepackage{setspace}

\makesavenoteenv{tabular}
\usepackage{tabularx}
\usepackage{booktabs}
\usepackage{multirow}
\usepackage{threeparttable}
\usepackage[referable]{threeparttablex} 
\newcolumntype{R}{>{\raggedleft\arraybackslash}X}
\newcolumntype{L}{>{\raggedright\arraybackslash}X}
\newcolumntype{C}{>{\centering\arraybackslash}X}
\newcolumntype{A}{>{\columncolor{gray!25}}C}
\newcolumntype{a}{>{\columncolor{gray!25}}c}

\usepackage{dcolumn} 
\newcolumntype{.}{D{.}{.}{-1}}

\usepackage{tikz}
\usetikzlibrary{arrows}
\usepackage[semicolon]{natbib}
\usepackage{hyperref} 
\usepackage{hyperref}
\hypersetup{
  colorlinks   = true,    
  urlcolor     = blue,    
  linkcolor    = blue,    
  citecolor    = red      
}

\usepackage{microtype}
\usepackage[justification=centerfirst]{caption}

\usepackage[labelformat=simple]{subcaption}

\DeclareCaptionLabelFormat{parenthesis}{(#2)}
\makeatletter
\renewcommand\p@subfigure{\arabic{figure}.}
\makeatother

\def\addlegendimage{\csname pgfplots@addlegendimage\endcsname}

\usepackage{enumitem}

\setlist[itemize]{leftmargin=3\parindent}
\setlist[enumerate]{leftmargin=2\parindent}

\theoremstyle{plain}

\newtheorem{corollary}{Corollary}
\newtheorem{lemma}{Lemma}

\newtheorem{proposition}{Proposition}
\newtheorem{theorem}{Theorem}

\theoremstyle{definition}

\newtheorem{definition}{Definition}
\newtheorem{example}{Example}

\theoremstyle{remark}

\newtheorem{remark}{Remark}

\def\keywords{\vspace{.5em} 
{\textit{Keywords}:\,\relax%
}}

\def\JEL{\vspace{.5em} 
{\textbf{\emph{JEL} classification number}:\,\relax%
}}

\def\AMS{\vspace{.5em} 
{\textbf{\emph{AMS} classification number}:\,}}

\author{L\'aszl\'o Csat\'o\thanks{e-mail: laszlo.csato@uni-corvinus.hu} }
\affil{Department of Operations Research and Actuarial Sciences \\ Corvinus University of Budapest \\ MTA-BCE ''Lend\"ulet'' Strategic Interactions Research Group \\ Budapest, Hungary}
\title{On the additivity of preference aggregation methods\thanks{~We thank S\'andor Boz\'oki for his comments on earlier versions of the manuscript. \newline
The research was supported by OTKA grant K-77420. \newline
This research was supported by the European Union and the State of Hungary, co-financed by the European Social Fund in the framework of TÁMOP 4.2.4. A/1-11-1-2012-0001 'National Excellence Program'.}}
\date{\today}

\begin{document}

\maketitle

\begin{abstract}
The paper reviews some axioms of additivity concerning ranking methods used for generalized tournaments with possible missing values and multiple comparisons. It is shown that one of the most natural properties, called consistency, has strong links to independence of irrelevant comparisons, an axiom judged unfavourable when players have different opponents.
Therefore some directions of weakening consistency are suggested, and several ranking methods, the score, generalized row sum and least squares as well as fair bets and its two variants (one of them entirely new) are analysed whether they satisfy the properties discussed. It turns out that least squares and generalized row sum with an appropriate parameter choice preserve the relative ranking of two objects if the ranking problems added have the same comparison structure.

\JEL{D71}

\AMS{15A06, 91B14}

\end{abstract}

\section{Introduction}

Paired-comparison based ranking emerges in many fields such as social choice theory \citep{ChebotarevShamis1998a}, sports \citep{Landau1895, Landau1914, Zermelo1929}, or psychology \citep{Thurstone1927}. Here the most general version of the problem, allowing for different preference intensities (including ties) as well as incomplete and multiple comparisons among the objects, is addressed.

The paper contributes to this field by the investigation of additivity: how the ranking changes by adding two independent tournaments. We get a certain impossibility theorem, either total additivity or independence of irrelevant comparisons should be sacrificed in order to get a meaningful ranking method. Therefore some directions of weakening additivity are studied. 

Due to the investigation of the performance of ranking methods with respect to the additive properties, the current paper can also be regarded as a supplement to the findings of \citet{ChebotarevShamis1998a} and \citet{Gonzalez-DiazHendrickxLohmann2013} by analysing new methods and axioms.

Throughout the paper, we concentrate on the scoring procedures listed below:
\begin{itemize}[label=$\bullet$]
\item
\textbf{Score}: a natural method for binary tournaments (for characterizations on restricted domains, see \citet{Young1974, HanssonSahlquist1976, Rubinstein1980, NitzanRubinstein1981, Bouyssou1992}).

\item
\textbf{Least squares}: a well-known procedure in statistics and psychology (see \citet{Thurstone1927, Gulliksen1956, KaiserSerlin1978}).

\item
\textbf{Generalized row sum}: a parametric family of ranking methods resulting in the score and least squares as limits (see \citet{Chebotarev1989, Chebotarev1994}).

\item
\textbf{Fair bets}: an extensively studied method in social choice theory as well as a procedure for ranking the nodes of directed graphs (see \citet{Daniels1969, MoonPullman1970, SlutzkiVolij2005, SlutzkiVolij2006, SlikkerBormvandenBrink2012}).

\item
\textbf{Dual fair bets}: a scoring procedure obtained from fair bets by 'reversing' an axiom in its characterization (see \citet{SlutzkiVolij2005}).

\item
\textbf{Copeland fair bets}: a novel method introduced in this paper by applying the idea of \citet{HeringsvanderLaanTalman2005} for fair bets.
\end{itemize}

A main, somewhat unexpected result is that one natural axiom of additivity, consistency -- which requires the relative ranking of two objects to remain the same if it agrees in both ranking problems -- seems to be a surprisingly severe condition.
First, among the procedures analysed, only the trivial score method satisfies it.
Second, together with two basic properties, it implies a kind of independence of irrelevant comparisons. However, the latter is a property one would rather not have in this general framework, since it means that the performance of the opponents (objects compared with a given one) does not count.

Therefore some directions of weakening additivity are studied. One of them turns out to be fruitful, at least in the case of some ranking procedures, which preserve the relative ranking when the ranking problems added have the same comparison structure. This axiom is worth to consider as a watershed, application of procedures without it remains dubious.

Another way to avoid the impossibility result is to restrict the domain, since independence of irrelevant comparisons does not cause problems in the case of round-robin tournaments. It will be revealed that fair bets, dual fair bets and Copeland fair bets show a strange behaviour even on this narrow subset.

The axiomatic approach followed offers some guidelines for the choice of the appropriate ranking procedure as well as it contributes to a better understanding of them.
It is important because, despite the extended literature (for reviews, see \citet{Laslier1997} and \citet{ChebotarevShamis1998a}), characterizations of scoring methods (which provide a ranking by associating scores for the objects such that a higher value corresponds to a better position in the ranking) on this wide domain are limited, they exist only for fair bets \citep{SlutzkiVolij2005} and invariant methods \citep{SlutzkiVolij2006}.

The paper is structured as follows. Section~\ref{Sec2} presents the setting of the problem, the definitions of ranking methods examined, and some invariance properties known from the literature. In Section~\ref{Sec3}, four axioms linked to additivity of ranking problems are reviewed. Section~\ref{Sec4} proves that the strongest additive property has unfavourable implications on the general domain used. Finally, Section~\ref{Sec5} concludes the results, summarizes them visually in a table, while the connections of the axioms are displayed in a graph.

\section{Preliminaries} \label{Sec2}

The following part of the paper discusses the representation of ranking problems, defines the scoring procedures investigated later, and presents some structural invariance axioms used in the literature.

\subsection{Notations} \label{Sec21}

Let $N = \{ X_1,X_2, \dots, X_n \}$, $n \in \mathbb{N}$ be the \emph{set of objects} and $T = (t_{ij}) \in \mathbb{R}^{n \times n}$ be the \emph{tournament matrix} such that $t_{ij} + t_{ji} \in \mathbb{N}$.
$t_{ij}$ represents the aggregate score of object $X_i$ against $X_j$, $t_{ij} / (t_{ij} + t_{ji})$ may be interpreted as the likelihood that object $X_i$ is better than object $X_j$. $t_{ii} = 0$ is assumed for all $i = 1,2, \dots ,n$.
A possible derivation of the tournament matrix can be found in \citet{Gonzalez-DiazHendrickxLohmann2013} and \citet{Csato2015a}.

The pair $(N,T)$ is called a \emph{ranking problem}.
The set of ranking problems is denoted by $\mathcal{R}$.
A \emph{scoring procedure} $f$ is an $\mathcal{R} \to \mathbb{R}^n$ function, giving a rating for each object. It immediately determines a ranking (a transitive and complete weak order on the set $N \times N$) $\succeq$ such that $f_i \geq f_j$ means that $X_i$ is ranked weakly above $X_j$, denoted by $X_i \succeq X_j$.
Ratings provide cardinal while rankings provide ordinal information about the objects.

\begin{remark}
Every scoring method can be considered as a \emph{ranking method}. This paper discusses only ranking methods derived from scoring procedures, the two notions will be used analogously.
\end{remark}

A ranking problem $(N,T)$ has the \emph{results matrix} $A = T - T^\top = (a_{ij}) \in \mathbb{R}^{n \times n}$ and the \emph{matches matrix} $M = T + T^\top = (m_{ij}) \in \mathbb{N}^{n \times n}$ such that $m_{ij}$ is the number of the comparisons between $X_i$ and $X_j$, whose outcome is given by $a_{ij}$. Matrices $A$ and $M$ also define the tournament matrix by $T = (A + M)/2$.

\begin{remark}
Note that any ranking problem $(N,T) \in \mathcal{R}$ can be denoted analogously as $(N,A,M)$ with the restriction $|a_{ij}| \leq m_{ij}$ for all $X_i,X_j \in N$, that is, the outcome of any paired comparison between two objects cannot 'exceed' their number of matches.
Despite it is not parsimonious, usually the second notation will be used in the following because it helps to define certain ranking methods and axioms.
\end{remark}

A ranking problem is called \emph{round-robin} if $m_{ij} = m$ for all $X_i \neq X_j$.
The set of round-robin ranking problems is denoted by $\mathcal{R}^R$.
$d_i = \sum_{j=1}^n m_{ij}$ is the \emph{total number of comparisons} of object $X_i$.
$m = \max_{X_i,X_j \in N} m_{ij}$ is the \emph{maximal number of comparisons} in the ranking problem.

Matrix $M$ can be represented by an undirected multigraph $G := (V,E)$, where vertex set $V$ corresponds to the object set $N$, and the number of edges between objects $X_i$ and $X_j$ is equal to $m_{ij}$. Then the degree of node $X_i$ is $d_i$.
Graph $G$ is the \emph{comparison multigraph} associated with the ranking problem $(N,A,M)$, however, it is independent of the results matrix $A$. The \emph{Laplacian matrix} $L = \left( \ell_{ij} \right) \in \mathbb{R}^{n \times n}$ of graph $G$ is given by $\ell_{ij} = -m_{ij}$ for all $X_i \neq X_j$ and $\ell_{ii} = d_i$ for all $X_i \in N$.

A \emph{path} from $X_{k_1}$ to $X_{k_s}$ is a sequence of objects $X_{k_1}, X_{k_2}, \dots , X_{k_s}$ such that $m_{k_\ell k_{\ell+1}} > 0$ for all $\ell = 1,2, \dots ,s-1$. Two objects are connected if there exists a path between them. Ranking problem $(N,A,M) \in \mathcal{R}$ is said to be \emph{connected} if every pair of objects is connected.
The set of connected ranking problems is denoted by $\mathcal{R}^C$.

A \emph{directed path} from $X_{k_1}$ to $X_{k_s}$ is a sequence of objects $X_{k_1}, X_{k_2}, \dots , X_{k_s}$ such that $t_{k_\ell k_{\ell+1}} > 0$ for all $\ell = 1,2, \dots ,s-1$.
Ranking problem $(N,T) \in \mathcal{R}$ is called \emph{irreducible} if there exists a directed path from $X_i$ to $X_j$ for all $X_i, X_j \in N$.
The set of irreducible ranking problems is denoted by $\mathcal{R}^I$.

Let $\mathbf{e} \in \mathbb{R}^n$ denote the column vector with $e_i = 1$ for all $i = 1,2, \dots ,n$. Let $I \in \mathbb{R}^{n \times n}$ be the identity matrix, $O \in \mathbb{R}^{n \times n}$ be the zero matrix.

\subsection{Ranking methods} \label{Sec22}

Tournament ranking involves three main challenges. The first one is the possible appearance of circular triads, when object $X_i$ is better than $X_j$ (that is, $a_{ij} > a_{ji}$), $X_j$ is better than $X_k$, but $X_k$ is better than $X_i$. If preference intensities also count as in the model above, other triplets ($X_i, X_j, X_k$) may produce problems, too.
The second problem is that the performance of objects compared with $X_i$ strongly influences the observable paired comparison outcomes $a_{ij}$. For example, if $X_i$ was compared only with $X_j$, then its rating may depend on other results of $X_j$.
The third difficulty is given by the different number of comparisons of the objects, $d_i \neq d_j$. It must be realized that there is no entirely satisfactory way of ranking if the number of replications of each object varies appreciably \citep[p.~1]{David1987}. However, the current paper does not deal with the question whether a given dataset may be globally ranked in a meaningful way or the data are inherently inconsistent, an issue investigated for example by \citet{JiangLimYaoYe2011}.
Since each problem occur just if $n \geq 3$, the case of two objects becomes trivial.

Now some scoring procedures are presented. They will be used only for ranking purposes, so they will be called ranking methods.
The first one does not take the comparison structure into account.

\begin{definition} \label{Def21}
\emph{Score}: $\mathbf{s}(N,A,M) = A \mathbf{e}$.
\end{definition}

The following \emph{parametric} procedure was constructed axiomatically by \citet{Chebotarev1989} and thoroughly analysed in \citet{Chebotarev1994}.

\begin{definition} \label{Def22}
\emph{Generalized row sum}: it is the unique solution $\mathbf{x}(\varepsilon)(N,A,M)$ of the system of linear equations $(I+ \varepsilon L) \mathbf{x}(\varepsilon)(N,A,M) = (1 + \varepsilon m n) \mathbf{s}$, where $\varepsilon > 0$ is a parameter. 
\end{definition}

Generalized row sum adjusts the standard score $s_i$ by accounting for the performance of objects compared with $X_i$, and adds an infinite depth to this argument: scores of all objects available on a path appear in the calculation. $\varepsilon$ indicates the importance attributed to this correction.
Generalized row sum results in score if $\varepsilon \to 0$.

\begin{lemma} \label{Lemma21}
$\lim_{\varepsilon \to 0} \mathbf{x}(\varepsilon)(N,A,M) = \mathbf{s}(N,A,M)$.
\end{lemma}

\begin{proof}
It follows from Definitions~\ref{Def21} and \ref{Def22}.
\end{proof}


Based on some reasonableness condition, \citet{Chebotarev1994} identifies a possible upper bound for $\varepsilon$.

\begin{definition} \label{Def23}
\emph{Reasonable choice of $\varepsilon$} \citep[Proposition~5.1]{Chebotarev1994}:
The \emph{reasonable upper bound} of $\varepsilon$ is $1 / \left[ m(n-2) \right]$.
\end{definition}

The reasonable choice is not well-defined in the trivial case of $n=2$, thus $n \geq 3$ is implicitly assumed in the following.

\begin{proposition} \label{Prop21}
If $\varepsilon$ is within the reasonable interval $\left( 0,\, 1 / \left[ m(n-2) \right] \right]$, then $-m(n-1) \leq x_i(\varepsilon)(N,A,M) \leq m(n-1)$ for all $X_i \in N$.
\end{proposition}

\begin{proof}
See \citet[Property~13]{Chebotarev1994}.
\end{proof}

Note that in a round-robin ranking problem $-m(n-1) \leq s_i(N,A,M) \leq m(n-1)$ holds for all $X_i \in N$. 

Both the score and generalized row sum rankings are well-defined and easily computable from a system of linear equations for all ranking problems $(N,A,M) \in \mathcal{R}$.

The least squares method was suggested by \citet{Thurstone1927} and \citet{Horst1932}.

\begin{definition} \label{Def24}
\emph{Least squares}: it is the solution $\mathbf{q}(N,A,M)$ of the system of linear equations $L \mathbf{q}(N,A,M) = \mathbf{s}(N,A,M)$ and $\mathbf{e}^\top \mathbf{q}(N,A,M) = 0$.
\end{definition}

Generalized row sum results in least squares if $\varepsilon \to \infty$.

\begin{lemma} \label{Lemma22}
$\lim_{\varepsilon \to \infty} \mathbf{x}(\varepsilon)(N,A,M) = mn \mathbf{q}(N,A,M)$.
\end{lemma}

\begin{proof}
It follows from Definitions~\ref{Def22} and \ref{Def24}.
\end{proof}

\begin{proposition} \label{Prop22}
The least squares ranking is unique if and only if the comparison multigraph $G$ of the ranking problem $(N,A,M) \in \mathcal{R}$ is connected.
\end{proposition}

\begin{proof}
See \citet{BozokiCsatoRonyaiTapolcai2015}.
\citet[p.~220]{ChebotarevShamis1999} mention this fact without further discussion.
\end{proof}

An extensive analysis and a graph interpretation, and further references can be found in \citet{Csato2015a}.


Several scoring procedures build upon the idea of rewarding wins without punishing losses. Two early contributions in this field are \citet{Wei1952} and \citet{Kendall1955}. They have been studied in social choice and game theory by \citet{BormvandenBrinkSlikker2002, HeringsvanderLaanTalman2005, SlikkerBormvandenBrink2012, SlutzkiVolij2005, SlutzkiVolij2006}, among others. 

One of the most widely used methods within this framework is the fair bets method, originally suggested by \citet{Daniels1969} and \citet{MoonPullman1970}.
This procedure was axiomatically characterized by \citet{SlutzkiVolij2005} and \citet{SlutzkiVolij2006}. Its properties have been investigated by \citet{Gonzalez-DiazHendrickxLohmann2013}.

Fair bets is defined with the notation $(N,T)$ for the sake of simplicity.
Let $F = \text{diag}(T^\top \mathbf{e})$, an $n \times n$ diagonal matrix showing the number of losses for each object. 

\begin{definition} \label{Def25}
\emph{Fair bets}: it is the solution $\mathbf{fb}(N,T)$ of the system of linear equations $F^{-1} T \mathbf{fb}(N,T) = \mathbf{fb}(N,T)$ and $\mathbf{e}^\top \mathbf{fb}(N,T) = 1$.
\end{definition}

\begin{proposition} \label{Prop23}
The fair bets ranking is unique if and only if the ranking problem $(N,T) \in \mathcal{R}$ is irreducible.
\end{proposition}

\begin{proof}
See \citet{MoonPullman1970}. 
\end{proof}

In the case of reducible ranking problems, Perron-Frobenius theorem does not guarantee that the eigenvector corresponding to the dominant eigenvalue is strictly positive.

Fair bets judges wins against better objects to be more important than losses against worse objects.
One may argue for the opposite, which implies the dual fair bets method \citep{SlutzkiVolij2005} using the transposed tournament matrix $T^\top$, but in this case a lower value is better.

\begin{definition} \label{Def26}
\emph{Dual fair bets}: it is $\mathbf{dfb}(N,T) = -\mathbf{dfb}^{\ast}(N,T)$, where $\mathbf{dfb}^{\ast}(N,T)$ is the solution of the system of linear equations $\left[ \text{diag}(T \mathbf{e}) \right]^{-1} T^\top \mathbf{dfb}^{\ast}(N,T) = \mathbf{dfb}^{\ast}(N,T)$ and $\mathbf{e}^\top \mathbf{dfb}^{\ast}(N,T) = 1$.
\end{definition}

The transformation $\mathbf{dfb}(N,T) = -\mathbf{dfb}^{\ast}(N,T)$ is necessary in order to ensure that $X_i \succeq X_j \Leftrightarrow dfb_i(N,T) \geq dfb_j(N,T)$ for all $X_i, X_j \in N$.

The axiomatization of fair bets also characterizes the dual fair bets by changing only one property, negative responsiveness to losses with positive responsiveness to wins \citep[Remark~1]{SlutzkiVolij2005}. These two approaches can be seen in the case of positional power, too, by the definition of positional power and positional weakness \citep{HeringsvanderLaanTalman2005}. Similarly to the their Copeland positional value, Copeland fair bets method is introduced as the sum of the fair bets and dual fair bets ratings. 

\begin{definition} \label{Def27}
\emph{Copeland fair bets}: $\mathbf{Cfb}(N,T) = \mathbf{fb}(N,T) + \mathbf{dfb}(N,T)$.
\end{definition}

Now $X_i \succeq X_j \Leftrightarrow Cfb_i(N,T) \geq Cfb_j(N,T)$ as earlier.

The six scoring procedures (Definitions \ref{Def21}-\ref{Def22} and \ref{Def24}-\ref{Def27}) are discussed with respect to their axiomatic properties.
\citet{Gonzalez-DiazHendrickxLohmann2013} have analysed the least squares and fair bets methods, as well as generalized row sum with the parameter $\varepsilon = 1/ \left[ m(n-2) \right]$. They use a different version of the score, $s_i / d_i$ for all $X_i \in N$.

Ranking problem $(N,A,M) \in \mathcal{R}$ can be represented by a graph such that the nodes are the objects, $k$ times $(X_i,X_j) \in N \times N$ undirected edge means $a_{ij} (= a_{ji}) = 0$, $m_{ij} = k$, and $k$ times $(X_i,X_j) \in N \times N$ directed edge means $k$ comparison with maximal intensity, that is, $a_{ij} = k$ ($a_{ji} = -k$), $m_{ij} = k$. We think it helps in understanding the examples.

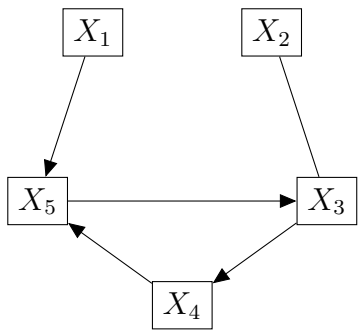
\begin{figure}[htbp]
\centering
\caption{Ranking problem of Example \ref{Examp1}}
\label{Fig1}
\begin{tikzpicture}[scale=1, auto=center, transform shape, >=triangle 45]
\tikzstyle{every node}=[draw,shape=rectangle];
  \node (n1) at (126:2)  {$X_1$};
  \node (n2) at (54:2)   {$X_2$};
  \node (n3) at (342:2)  {$X_3$};
  \node (n4) at (270:2)  {$X_4$};
  \node (n5) at (198:2)  {$X_5$};

  \foreach \from/\to in {n1/n5,n3/n4,n4/n5,n5/n3}
    \draw [->] (\from) -- (\to);
    
  \draw (n2) -- (n3);
\end{tikzpicture}
\end{figure}

\begin{example} \label{Examp1}
\citep[Example~2]{Chebotarev1994} Let $(N,A,M) \in \mathcal{R}$ be the ranking problem in Figure \ref{Fig1} with the set of objects $N = \{ X_1,X_2,X_3,X_4,X_5 \}$.

The corresponding tournament, results and matches matrices are as follows
\[
T =
\begin{pmatrix}
    0     & 0     & 0     & 0     & 1  \\
    0     & 0     & 0.5   & 0     & 0  \\
    0     & 0.5   & 0     & 1     & 0  \\
    0     & 0     & 0     & 0     & 1  \\
    0     & 0     & 1     & 0     & 0  \\
\end{pmatrix},
\quad
A =
\begin{pmatrix}
    0     & 0     & 0     & 0     & 1  \\
    0     & 0     & 0     & 0     & 0  \\
    0     & 0     & 0     & 1     & -1 \\
    0     & 0     & -1    & 0     & 1  \\
    -1    & 0     & 1     & -1    & 0  \\
\end{pmatrix},
\quad
M =
\begin{pmatrix}
    0     & 0     & 0     & 0     & 1 \\
    0     & 0     & 1     & 0     & 0 \\
    0     & 1     & 0     & 1     & 1 \\
    0     & 0     & 1     & 0     & 1 \\
    1     & 0     & 1     & 1     & 0 \\
\end{pmatrix}.
\]

\begin{table}[htbp]
\centering
\caption{Generalized row sum vectors $\mathbf{x}(\varepsilon)$ of Example~\ref{Examp1}}
\label{Table1}
\begin{footnotesize}
    \begin{tabularx}{\textwidth}{c RRRRRRR} \toprule
    $\varepsilon$ & \multicolumn{1}{c}{$0$}     & \multicolumn{1}{c}{$1/100$}  & \multicolumn{1}{c}{$1/4$}   & \multicolumn{1}{c}{$1/3$}   & \multicolumn{1}{c}{$1$}   & \multicolumn{1}{c}{$5$}     & \multicolumn{1}{c}{$\to \infty$} \\
    \midrule
    $X_1$ & $1.0000$ & $1.0296$ & $1.7165$ & $2.2649$ & $2.4242$ & $3.4369$ & $4.0000$ \\
    $X_2$ & $0.0000$ & $-0.0001$ & $-0.0613$ & $-0.1917$ & $-0.2424$ & $-0.6819$ & $-1.0000$ \\
    $X_3$ & $0.0000$ & $-0.0099$ & $-0.2452$ & $-0.4314$ & $-0.4848$ & $-0.8183$ & $-1.0000$ \\
    $X_4$ & $0.0000$ & $-0.0100$ & $-0.2759$ & $-0.4878$ & $-0.5455$ & $-0.8609$ & $-1.0000$ \\
    $X_5$ & $-1.0000$ & $-1.0096$ & $-1.1341$ & $-1.1540$ & $-1.1515$ & $-1.0757$ & $-1.0000$ \\
    \bottomrule
    \end{tabularx} 
\end{footnotesize}
\end{table}

The solutions with generalized row sum for various values of $\varepsilon$ are given in Table~\ref{Table1}.
Here $m=1$ and $n=5$, thus $\varepsilon = 1/3$ is the reasonable upper bound by Definition~\ref{Def23}. The ranking of the objects is $X_1 \succ X_2 \succ X_3 \succ X_4 \succ X_5$ for all positive parameters since $X_1$ dominates $X_5$, which effects $X_3$ and $X_4$ through the circular triad $(X_3, X_4, X_5)$. However, $X_3$ has a draw against $X_2$. Note that $X_2 \sim X_3 \sim X_4$ for the score ($\varepsilon \to 0$) and least squares methods ($\varepsilon \to \infty$), referring to a kind of neglect of the comparison between $X_2$ and $X_3$.

Example~\ref{Examp1} is an irreducible ranking problem, so fair bets rating is not unique. Nevertheless, a ranking can be obtained by the application of its extension according to \citet{SlutzkiVolij2005}: $X_1$ is the best object as no other has any chance to defeat it, and the remaining four form an irreducible component. It results in $X_1 \succ (X_2 \sim X_3 \sim X_4 \sim X_5)$, which coincides with the one from least squares. Similarly, both dual fair bets and Copeland fair bets give $X_1 \succ (X_2 \sim X_3 \sim X_4 \sim X_5)$.
\end{example}

Because of Propositions~\ref{Prop22} and \ref{Prop23}, we restrict our analysis to the class of connected ranking problems $\mathcal{R}^C$, and to the set of irreducible ranking problems $\mathcal{R}^I$ in the case of fair bets. In ranking problems without a connected comparison multigraph, the rating of all objects on a common scale seems to be arbitrary.

\subsection{Structural invariance properties} \label{Sec23}

The main discussion requires the knowledge of some basic axioms already introduced.


\begin{definition} \label{Def31}
\emph{Neutrality} ($NEU$) \citep{Young1974}:
Let $(N,A,M) \in \mathcal{R}$ be a ranking problem and $\sigma: N \rightarrow N$ be a permutation on the set of objects. Let $\sigma(N,A,M) \in \mathcal{R}$ be the ranking problem obtained from $(N,A,M)$ by this permutation.
Scoring method $f: \mathcal{R} \to \mathbb{R}^n$ is \emph{neutral} if $f_i(N,A,M) \geq f_j(N,A,M) \Leftrightarrow f_{\sigma i} \left[ \sigma(N,A,M) \right] \geq f_{\sigma j} \left[ \sigma(N,A,M) \right]$ holds for all $X_i,X_j \in N$.
\end{definition}

Neutrality is a simple independence of labelling of the objects, and was called \emph{anonymity} in \citet{Bouyssou1992, SlutzkiVolij2005, Gonzalez-DiazHendrickxLohmann2013}. It is equivalent to the requirement that the permutation of two objects do not affect the ranking.

\begin{remark} \label{Rem2}
Let $f: \mathcal{R} \to \mathbb{R}^n$ be a neutral scoring procedure. If for the objects $X_i, X_j \in N$, $m_{ij} = 0$, and $a_{ik} = a_{jk}$, $m_{ik} = m_{jk}$ hold for all $X_k \in N \setminus \{ X_i,X_j \}$, then $f_i(N,A,M) = f_j(N,A,M)$ \citep[p.~62]{Bouyssou1992}.
\end{remark}

Remark~\ref{Rem2} claims that two indistinguishable objects have the same rank.

\begin{lemma} \label{Lemma31}
All methods presented above satisfy $NEU$.
\end{lemma}

\begin{proof}
It follows from their definitions.
\end{proof}

\begin{definition} \label{Def32}
\emph{Symmetry} ($SYM$) \citep{Gonzalez-DiazHendrickxLohmann2013}:
Let $(N,A,M) \in \mathcal{R}$ be a ranking problem such that $A = O$.
Scoring method $f: \mathcal{R} \to \mathbb{R}^n$ is \emph{symmetric} if $f_i(N,A,M) = f_j(N,A,M)$ for all $X_i, X_j \in N$.
\end{definition}

Symmetry does not require that objects $X_i$ and $X_j$ have the same number of comparisons ($d_i = d_j$). \citet{Young1974} and \citet[Axiom~4]{NitzanRubinstein1981} have introduced the property \emph{cancellation} for round-robin ranking problems, which coincides with symmetry on this set.

\begin{lemma} \label{Lemma32}
All methods presented above satisfy $SYM$.
\end{lemma}

\begin{proof}
It follows from their definitions.
\end{proof}

\begin{definition} \label{Def33}
\emph{Inversion} ($INV$) \citep{ChebotarevShamis1998a}:
Let $(N,A,M) \in \mathcal{R}$ be a ranking problem.
Scoring method $f: \mathcal{R} \to \mathbb{R}^n$ is \emph{invertible} if $f_i(N,A,M) \geq f_j(N,A,M) \Leftrightarrow f_i(N,-A,M) \leq f_j(N,-A,M)$ for all $X_i, X_j \in N$.
\end{definition}

Inversion means that taking the opposite of all results changes the ranking accordingly. It establishes a uniform treatment of victories and defeats.

\begin{remark} \label{Rem3}
Let $f: \mathcal{R} \to \mathbb{R}^n$ be a scoring procedure satisfying $INV$. Then $f_i(N,A,M) > f_j(N,A,M) \Leftrightarrow f_i(N,-A,M) < f_j(N,-A,M)$ for all $X_i, X_j \in N$.
\end{remark}

The following result was mentioned by \citet[p.~150]{Gonzalez-DiazHendrickxLohmann2013}.

\begin{corollary} \label{Col1}
$INV$ implies $SYM$.
\end{corollary}

\begin{lemma} \label{Lemma33}
The score, generalized row sum and least squares methods satisfy $INV$.
\end{lemma}

\begin{proof}
It is an immediate consequence of $\mathbf{s}(N,-A,M) = - \mathbf{s}(N,A,M)$.
\end{proof}

\begin{lemma} \label{Lemma34}
Fair bets and dual fair bets methods do not satisfy $INV$ even on the set $\mathcal{R}^R$.
\end{lemma}

\begin{proof}
See \citet[Example~4.4]{Gonzalez-DiazHendrickxLohmann2013} for fair bets. The same counterexample with a transposed tournament matrix proves the statement for dual fair bets. 
\end{proof}

Fair bets and dual fair bets violate inversion because of the different treatment of victories and losses. The potential problem can be seen still on the most simple domain of round-robin ranking problems.
However, their appropriate aggregation eliminates this strange feature, the major weakness of fair bets according to \citet[p.~164]{Gonzalez-DiazHendrickxLohmann2013}.

\begin{lemma} \label{Lemma35}
Copeland fair bets satisfies $INV$.
\end{lemma}

\begin{proof}
Consider the ranking problems $(N,T)$ and $(N,T^\top)$. $\mathbf{Cfb}(N,T) = \mathbf{fb}(N,T) + \mathbf{dfb}(N,T) = -\mathbf{dfb}(N,T^\top) - \mathbf{fb}(N,T^\top) = -\mathbf{Cfb}(N,T^\top)$.
\end{proof}


\section{Axioms of additivity} \label{Sec3}

This section reviews axioms of additivity, that is, the implications of summing two ranking problems for the ranking. Two new properties will be introduced in the wake of two known requirements. The restricted domain of round-robin ranking problems will be investigated, too.

\subsection{Properties already introduced} \label{Sec31}

As a first step some results of the existing literature is collected and refined.

\begin{definition} \label{Def51}
\emph{Consistency} ($CS$) \citep{Young1974}:
Let $(N,A,M),(N,A',M') \in \mathcal{R}$ be two ranking problems and $X_i, X_j \in N$ be two objects.
Let $f: \mathcal{R} \to \mathbb{R}^n$ be a scoring procedure such that $f_i(N,A,M) \geq f_j(N,A,M)$ and $f_i(N,A',M') \geq f_j(N,A',M')$.
$f$ is called \emph{consistent} if $f_i(N,A+A',M+M') \geq f_j(N,A+A',M+M')$, furthermore, $f_i(N,A+A',M+M') > f_j(N,A+A',M+M')$ if $f_i(N,A,M) > f_j(N,A,M)$ or $f_i(N,A',M') > f_j(N,A',M')$.
\end{definition}

$CS$ is the most general and intuitive version of additivity: if $X_i$ is not worse than $X_j$ in both ranking problems, this should not change after adding them up.
\citet{Young1974} used it only in the case of round-robin tournaments.

\begin{lemma} \label{Lemma51}
The score method satisfies $CS$.
\end{lemma}

\begin{proof}
It follows from Definition~\ref{Def21}.
\end{proof}

\begin{proposition} \label{Prop51}
The generalized row sum and least squares methods violate $CS$.
\end{proposition}

\citet[Example~4.2]{Gonzalez-DiazHendrickxLohmann2013} have shown the violation of a weaker property called \emph{order preservation} for the least squares and generalized row sum with $\varepsilon = 1/ \left[ m (n-2) \right]$.\footnote{~Order preservation contains the further requirement of $d_i / d_j = d_i' / d_j'$ for all $X_i, X_j \in N$, that is, the ratio of the matches is equal in the ranking problems added.}

\begin{proof}

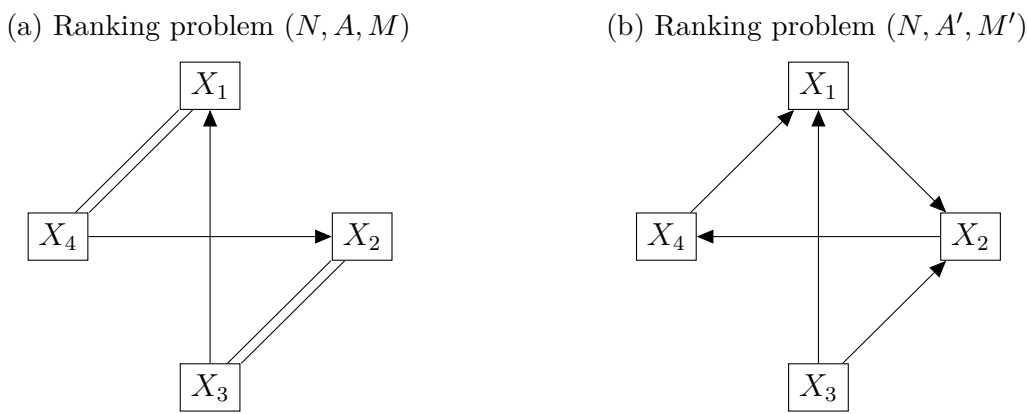
\begin{figure}[htbp]
\centering
\caption{Ranking problems of Example~\ref{Examp5}}
\label{Fig5}
  
\begin{subfigure}{.5\textwidth}
  \centering
  \caption{Ranking problem $(N,A,M)$}
  \label{Fig5a}
\begin{tikzpicture}[scale=1, auto=center, transform shape, >=triangle 45]
\tikzstyle{every node}=[draw,shape=rectangle];
  \node (n1) at (90:2) {$X_1$};
  \node (n2) at (0:2)  {$X_2$};
  \node (n3) at (270:2) {$X_3$};
  \node (n4) at (180:2)  {$X_4$};

  \foreach \from/ \to in {n3/n1,n4/n2}
    \draw [->] (\from) -- (\to);
    
\draw[transform canvas={xshift=0.5ex}](n1) -- (n4);
\draw[transform canvas={xshift=-0.5ex}](n1) -- (n4);
\draw[transform canvas={xshift=0.5ex}](n2) -- (n3);
\draw[transform canvas={xshift=-0.5ex}](n2) -- (n3);
\end{tikzpicture}
\end{subfigure}
\begin{subfigure}{.5\textwidth}
  \centering
  \caption{Ranking problem $(N,A',M')$}
  \label{Fig5b}
\begin{tikzpicture}[scale=1, auto=center, transform shape, >=triangle 45]
\tikzstyle{every node}=[draw,shape=rectangle];
  \node (n1) at (90:2) {$X_1$};
  \node (n2) at (0:2)  {$X_2$};
  \node (n3) at (270:2) {$X_3$};
  \node (n4) at (180:2)  {$X_4$};

  \foreach \from/ \to in {n1/n2,n2/n4,n3/n1,n3/n2,n4/n1}
    \draw [->] (\from) -- (\to);
\end{tikzpicture}
\end{subfigure}
\end{figure}

\begin{example} \label{Examp5}
Let $(N,A,M),(N,A',M') \in \mathcal{R}$ be the ranking problems in Figure \ref{Fig5} with the set of objects $N = \{ X_1,X_2,X_3,X_4 \}$ and tournament matrices
\[
T =
\begin{pmatrix}
    0     & 0     & 0     & 1 \\
    0     & 0     & 1     & 0 \\
    1     & 1     & 0     & 0 \\
    1     & 1     & 0     & 0 \\
\end{pmatrix} \qquad \text{and} \qquad
T' =
\begin{pmatrix}
    0     & 1     & 0     & 0 \\
    0     & 0     & 0     & 1 \\
    1     & 1     & 0     & 0 \\
    1     & 0     & 0     & 0 \\
\end{pmatrix}.
\]
Let $(N,A'',M'') = (N,A+A',M+M') \in \mathcal{R}$ be the sum of these two ranking problems.
\end{example}

Let $\mathbf{x}(\varepsilon)(N,A,M) = \mathbf{x}(\varepsilon)$, $\mathbf{x}(\varepsilon)(N,A',M') = \mathbf{x}(\varepsilon)'$, $\mathbf{x}(\varepsilon)(N,A'',M'') = \mathbf{x}(\varepsilon)''$ and $\mathbf{q}(N,A,M) = \mathbf{q}$, $\mathbf{q}(N,A',M') = \mathbf{q}'$, $\mathbf{q}(N,A'',M'') = \mathbf{q}''$.
Now $n=4$, $m = 2$, $m' = 1$, and $m'' = 3$. Therefore
\[
x_1(\varepsilon) = x_2(\varepsilon) = - \frac{1 + 14 \varepsilon + 56 \varepsilon^2 + 64 \varepsilon^3}{1 + 12 \varepsilon + 44 \varepsilon^2 + 48 \varepsilon^3} \text{, and}
\]
\[
x_1(\varepsilon)' = x_2(\varepsilon)' = -1 \text{, but}
\]
\[
x_1(\varepsilon)'' - x_2(\varepsilon)'' = - \frac{2 \varepsilon + 44 \varepsilon^2 + 240 \varepsilon^3}{1 + 22 \varepsilon + 154 \varepsilon^2 + 340 \varepsilon^3} < 0.
\]
It implies that $X_1 \sim_{(N,A,M)}^{\mathbf{x}(\varepsilon)} X_2$ and $X_1 \sim_{(N,A',M')}^{\mathbf{x}(\varepsilon)} X_2$, however, $X_1 \prec_{(N,A'',M'')}^{\mathbf{x}(\varepsilon)} X_2$. Generalized row sum is not consistent for any $\varepsilon$.

For the least squares method on the basis of Lemma~\ref{Lemma22}:
\[
q_1 = \frac{\lim_{\varepsilon \to \infty} x_1(\varepsilon)}{m n} = -\frac{64}{48} \cdot \frac{1}{2 \cdot 4}  = - \frac{1}{6} = \frac{\lim_{\varepsilon \to \infty} x_2(\varepsilon)}{m n} = q_2 \text{, and}
\]
\[
q_1' = \frac{\lim_{\varepsilon \to \infty} x_1(\varepsilon)'}{m' n} = -\frac{1}{4} = \frac{\lim_{\varepsilon \to \infty} x_2(\varepsilon)'}{m' n} = q_2' \text{, but }
\]
\[
q_1'' - q_2'' = \frac{\lim_{\varepsilon \to \infty} \left[ x_1(\varepsilon)'' - x_2(\varepsilon)'' \right]}{m'' n} = - \frac{240}{340} \cdot \frac{1}{3 \cdot 4} = - \frac{1}{17} < 0.
\]
Hence $X_1 \sim_{(N,A,M)}^{\mathbf{q}} X_2$ and $X_1 \sim_{(N,A',M')}^{\mathbf{q}} X_2$, but $X_1 \prec_{(N,A'',M'')}^{\mathbf{q}} X_2$.
\end{proof}

We will return later to the examination of fair bets and connected methods.

\citet{Gonzalez-DiazHendrickxLohmann2013} also discusses the following, strongly restricted version of additivity.

\begin{definition} \label{Def52}
\emph{Flatness preservation} ($FP$) \citep{SlutzkiVolij2005}:
Let $(N,A,M), \linebreak (N,A',M') \in \mathcal{R}$ be two ranking problems.
Let $f: \mathcal{R} \to \mathbb{R}^n$ be a scoring procedure such that $f_i(N,A,M) = f_j(N,A,M)$ and $f_i(N,A',M') = f_j(N,A',M')$ for all $X_i,X_j \in N$.
$f$ \emph{preserves flatness} if $f_i(N,A+A',M+M') = f_j(N,A+A',M+M')$ for all $X_i, X_j \in N$.
\end{definition}

$FP$ demands additivity only for problems where all objects are ranked uniformly. It is used by \citet{SlutzkiVolij2005} for the characterization of fair bets.

\begin{corollary} \label{Col2}
$CS$ implies $FP$.
\end{corollary}

\begin{proof}
It follows from Definitions~\ref{Def51} and \ref{Def52}.
\end{proof}

\begin{lemma} \label{Lemma52}
The score, generalized row sum and least squares methods satisfy $FP$.
\end{lemma}

It had been shown in \citet[Corollary~4.3]{Gonzalez-DiazHendrickxLohmann2013} for the least squares, and  in \citet[Proposition~4.2]{Gonzalez-DiazHendrickxLohmann2013} for generalized row sum with $\varepsilon = 1 / \left[ m(n-2) \right]$.

\begin{proof}
The score method preserves flatness due to Lemma~\ref{Lemma51} and Corollary~\ref{Col2}.

If $x_i(\varepsilon)(N,A,M) = x_j(\varepsilon)(N,A,M)$ for all $X_i,X_j \in N$, then $\mathbf{x}(\varepsilon)(N,A,M) = \mathbf{0}$.
We prove that $\mathbf{s}(N,A,M) = \mathbf{0} \Leftrightarrow \mathbf{x}(\varepsilon)(N,A,M) = \mathbf{0}$.
$s_i(N,A,M) = s_j(N,A,M)$ for all $X_i,X_j \in N$ implies $\mathbf{s}(N,A,M) = \mathbf{0}$, therefore $\mathbf{x}(\varepsilon)(N,A,M) = \mathbf{0}$.
On the other hand, $\mathbf{x}(\varepsilon)(N,A,M) = \mathbf{0}$ implies $(1 + \varepsilon m n) \mathbf{s}(N,A,M) = \mathbf{0}$, so $\mathbf{s}(N,A,M) = \mathbf{0}$.

The same argument can be applied in the case of least squares.
\end{proof}

\begin{lemma} \label{Lemma53}
Fair bets, dual fair bets and Copeland fair bets methods satisfy $FP$.
\end{lemma}

\begin{proof}
See \citet[Theorem~1]{SlutzkiVolij2005} for the fair bets. According to \citet[Remark~1]{SlutzkiVolij2005}, it is true for dual fair bets, too. It implies that Copeland fair bets also preserves flatness.
\end{proof}

To conclude, among the ranking procedures discussed, only the score method satisfies the strongest possible version of additivity (it will be shown later that fair bets and its peers breaak consistency). However, all of them meets an almost trivial property called flatness preservation.
It remains to be seen how they behave between these extremities.

\subsection{Two new requirements} \label{Sec32}

All objects ranked uniformly seems to be a tough condition in $FP$, therefore it makes sense to require additivity on a larger set. An obvious choice can be that only the objects involved are ranked equally.

\begin{definition} \label{Def53}
\emph{Equality preservation} ($EP$):
Let $(N,A,M), (N,A',M') \in \mathcal{R}$ be two ranking problems and $X_i, X_j \in N$ be two objects.
Let $f: \mathcal{R} \to \mathbb{R}^n$ be a scoring procedure such that $f_i(N,A,M) = f_j(N,A,M)$ and $f_i(N,A',M') = f_j(N,A',M')$.
$f$ \emph{preserves equality} if $f_i(N,A+A',M+M') = f_j(N,A+A',M+M')$.
\end{definition}


\begin{corollary} \label{Col3}
$CS$ implies $EP$. \\
$EP$ implies $FP$.
\end{corollary}

\begin{proof}
It follows from Definitions~\ref{Def51} and \ref{Def53}, and Definitions~\ref{Def52} and \ref{Def53}, respectively.
\end{proof}

\begin{lemma} \label{Lemma54}
The score method satisfies $EP$.
\end{lemma}

\begin{proof}
It comes from Lemma~\ref{Lemma51} and Corollary~\ref{Col3}.
\end{proof}

\begin{lemma} \label{Lemma55}
The generalized row sum and least squares methods violate $EP$.
\end{lemma}

\begin{proof}
In Example~\ref{Examp5}, $X_1 \sim_{(N,A,M)}^{\mathbf{x}(\varepsilon)} X_2$ and $X_1 \sim_{(N,A,M)}^{\mathbf{q}} X_2$ as well as $X_1 \sim_{(N,A',M')}^{\mathbf{x}(\varepsilon)} X_2$ and $X_1 \sim_{(N,A',M')}^{\mathbf{q}} X_2$, but $X_1 \prec_{(N,A'',M'')}^{\mathbf{x}(\varepsilon} X_2$ and $X_1 \prec_{(N,A'',M'')}^{\mathbf{q}} X_2$.
\end{proof}

\begin{proposition} \label{Prop52}
Fair bets, dual fair bets and Copeland fair bets methods violate $EP$.
\end{proposition}

\begin{proof}

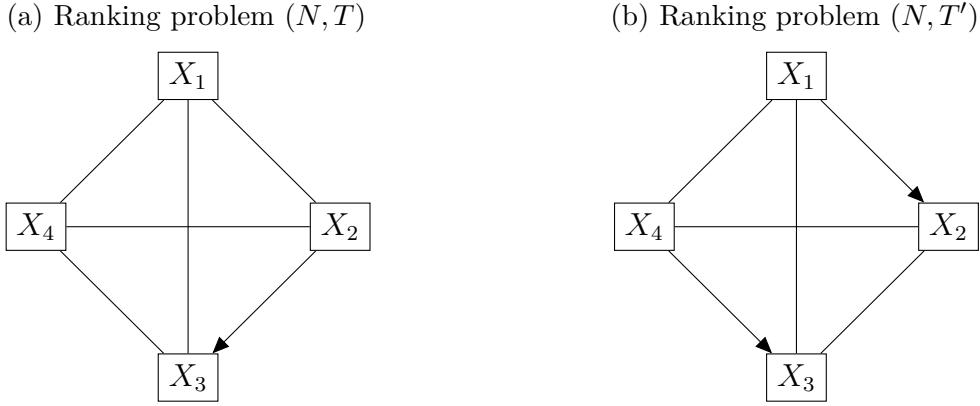
\begin{figure}[htbp]
\centering
\caption{Ranking problems of Example~\ref{Examp6}}
\label{Fig6}
  
\begin{subfigure}{.5\textwidth}
  \centering
  \caption{Ranking problem $(N,T)$}
  \label{Fig6a}
\begin{tikzpicture}[scale=1, auto=center, transform shape, >=triangle 45]
\tikzstyle{every node}=[draw,shape=rectangle];
  \node (n1) at (90:2) {$X_1$};
  \node (n2) at (0:2)  {$X_2$};
  \node (n3) at (270:2) {$X_3$};
  \node (n4) at (180:2)  {$X_4$};

\draw(n1) -- (n2);
\draw(n1) -- (n3);
\draw(n1) -- (n4);
\draw[->](n2) -- (n3);
\draw(n2) -- (n4);
\draw(n3) -- (n4);
\end{tikzpicture}
\end{subfigure}
\begin{subfigure}{.5\textwidth}
  \centering
  \caption{Ranking problem $(N,T')$}
  \label{Fig6b}
\begin{tikzpicture}[scale=1, auto=center, transform shape, >=triangle 45]
\tikzstyle{every node}=[draw,shape=rectangle];
  \node (n1) at (90:2) {$X_1$};
  \node (n2) at (0:2)  {$X_2$};
  \node (n3) at (270:2) {$X_3$};
  \node (n4) at (180:2)  {$X_4$};

\draw[->](n1) -- (n2);
\draw(n1) -- (n3);
\draw(n1) -- (n4);
\draw(n2) -- (n3);
\draw(n2) -- (n4);
\draw[->](n4) -- (n3);
\end{tikzpicture}
\end{subfigure}
\end{figure}

\begin{example} \label{Examp6}
Let $(N,T),(N,T') \in \mathcal{R}$ be the ranking problems in Figure \ref{Fig6} with the set of objects $N = \{ X_1,X_2,X_3,X_4 \}$ and tournament matrices
\[
T =
\begin{pmatrix}
    0     & 0.5   & 0.5   & 0.5 \\
    0.5   & 0     & 1     & 0.5 \\
    0.5   & 0     & 0     & 0.5 \\
    0.5   & 0.5   & 0.5   & 0 \\
\end{pmatrix} \qquad \text{and} \qquad
T' =
\begin{pmatrix}
    0     & 1     & 0.5   & 0.5 \\
    0     & 0     & 0.5   & 0.5 \\
    0.5   & 0.5   & 0     & 0 \\
    0.5   & 0.5   & 1     & 0 \\
\end{pmatrix}.
\]
Let $(N,T'') = (N,T+T') \in \mathcal{R}$ be the sum of these two ranking problems.
\end{example}

\begin{table}[htbp]
\centering
\caption{Fair bets and associated rating vectors of Example~\ref{Examp6}}
\label{Table4}
\begin{footnotesize}
\noindent\makebox[\textwidth]{
    \begin{tabularx}{1\textwidth}{C rrr rrr rrr} \toprule
    & $\mathbf{fb}(T)$     & $\mathbf{dfb}(T)$  & $\mathbf{Cfb}(T)$   & $\mathbf{fb}(T')$   & $\mathbf{dfb}(T')$   & $\mathbf{Cfb}(T')$ & $\mathbf{fb}(T'')$   & $\mathbf{dfb}(T'')$   & $\mathbf{Cfb}(T'')$ \\
\midrule
    $X_1$  & $1/4$  & $-1/4$  & $0$     & $3/8$  & $-1/8$  & $1/4$   & $163/512$ & $-101/512$ & $31/256$ \\
    $X_2$  & $3/8$  & $-1/8$  & $1/4$   & $1/8$  & $-3/8$  & $-1/4$  & $117/512$ & $-115/512$ & $1/256$ \\
    $X_3$  & $1/8$  & $-3/8$  & $-1/4$  & $1/8$  & $-3/8$  & $-1/4$  & $75/512$  & $-205/512$ & $-65/256$ \\
    $X_4$  & $1/4$  & $-1/4$  & $0$     & $3/8$  & $-1/8$  & $1/4$   & $157/512$ & $-91/512$  & $33/256$ \\
    \bottomrule
    \end{tabularx} }
\end{footnotesize}
\end{table}

The rating vectors are given in Table~\ref{Table4}: $X_1 \sim_{(N,T)} X_4$ and $X_1 \sim_{(N,T')} X_4$ for the three methods, but $X_1 \succ_{(N,T'')}^{\mathbf{fb}} X_4$, $X_1 \prec_{(N,T'')}^{\mathbf{dfb}} X_4$, and $X_1 \prec_{(N,T'')}^{\mathbf{Cfb}} X_4$.
\end{proof}

\begin{lemma} \label{Lemma56}
Fair bets, dual fair bets and Copeland fair bets methods violate $CS$.
\end{lemma}

\begin{proof}
It comes from Proposition~\ref{Prop52} and Corollary~\ref{Col3}.
\end{proof}

Another obvious restriction on $CS$ can be to allow only for the combination of ranking problems with the same matches matrix, when the interaction of different comparison multigraphs is eliminated.

\begin{definition} \label{Def54}
\emph{Result consistency} ($RCS$):
Let $(N,A,M), (N,A',M) \in \mathcal{R}$ be two ranking problems and $X_i, X_j \in N$ be two objects.
Let $f: \mathcal{R} \to \mathbb{R}^n$ be a scoring procedure such that $f_i(N,A,M) \geq f_j(N,A,M)$ and $f_i(N,A',M) \geq f_j(N,A',M)$.
$f$ is called \emph{result consistent} if $f_i(N,A+A',2M) \geq f_j(N,A+A',2M)$, furthermore, $f_i(N,A+A',2M) > f_j(N,A+A',2M)$ if $f_i(N,A,M) > f_j(N,A,M)$ or $f_i(N,A',M) > f_j(N,A',M)$.
\end{definition}

\begin{corollary} \label{Col4}
$CS$ implies $RCS$.
\end{corollary}

\begin{proof}
It follows from Definitions~\ref{Def51} and \ref{Def54}.
\end{proof}

\begin{proposition} \label{Prop53}
$RCS$ and $SYM$ imply $INV$.
\end{proposition}

\begin{proof}
Consider a ranking problem $(N,A,M) \in \mathcal{R}$ with $f_i(N,A,M) \geq f_j(N,A,M)$ for objects $X_i,X_j \in N$. 
If $f_i(N,-A,M) > f_j(N,-A,M)$, then $f_i(N,O,2M) > f_j(N,O,2M)$ due to $RCS$, which contradicts to $SYM$. Therefore $f_i(N,-A,M) \leq f_j(N,-A,M)$.
\end{proof}

\begin{corollary} \label{Col6}
$CS$ and $SYM$ imply $INV$.
\end{corollary}

\begin{proof}
It follows from Proposition~\ref{Prop53} and Corolllary~\ref{Col4}.
\end{proof}

Corollary~\ref{Col6} was proved by \citet[Lemma~1]{NitzanRubinstein1981} in the case of round-robin ranking problems (on the set $\mathcal{R}^R$), when $CS$ is equivalent to $RCS$ and $SYM$ is an almost trivial condition.

\begin{lemma} \label{Lemma57}
The score method satisfies $RCS$.
\end{lemma}

\begin{proof}
It can be derived from Lemma~\ref{Lemma51} and Corollary~\ref{Col4}.
\end{proof}

\begin{proposition} \label{Prop54}
The least squares method satisfies $RCS$.
\end{proposition}

\begin{proof}
Let $\mathbf{q}(N,A,M) = \mathbf{q}$, $\mathbf{q}(N,A',M) = \mathbf{q}'$ and $\mathbf{q}(N,A+A',M+M) = \mathbf{q}''$.
It is shown that $2 \mathbf{q}'' = \mathbf{q} + \mathbf{q}'$. The Laplacian matrix of the comparison multigraph associated with matches matrix $2M$ is $2L$, so
\[
2L \mathbf{q}'' = \mathbf{s}(N,A+A',M+M) = \mathbf{s}(N,A,M) + \mathbf{s}(N,A',M) = L \left( \mathbf{q} + \mathbf{q}' \right)
\]
as well as $\mathbf{e}^\top \mathbf{q}'' = \mathbf{e}^\top \left[ (1/2) \mathbf{q} + (1/2) \mathbf{q}' \right] = 0$.
\end{proof}

Regarding the generalized row sum, two cases should be distinguished by the parameter choice.

\begin{proposition} \label{Lemma58}
The generalized row sum method with a fixed $\varepsilon$ may violate $RCS$.
\end{proposition}

\begin{proof}
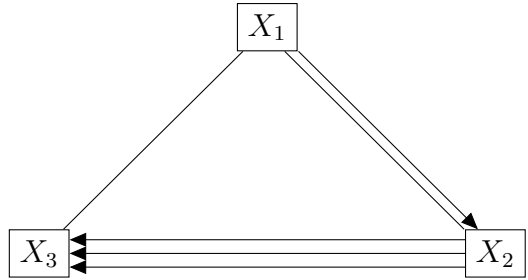
\begin{figure}[htbp]
\centering
\caption{Ranking problem of Example~\ref{Examp2}}
\label{Fig2}
  
\begin{tikzpicture}[scale=1, auto=center, transform shape, >=triangle 45]
\tikzstyle{every node}=[draw,shape=rectangle];
  \node (n1) at (90:3) {$X_1$};
  \node (n2) at (0:3) {$X_2$};
  \node (n3) at (180:3) {$X_3$};

\draw[transform canvas={xshift=-0.5ex}](n1) -- (n2);
\draw[->,transform canvas={xshift=0.5ex}](n1) -- (n2);
\draw(n1) -- (n3);
\draw[->,transform canvas={yshift=1ex}](n2) -- (n3);
\draw[->](n2) -- (n3);
\draw[->,transform canvas={yshift=-1ex}](n2) -- (n3);
\end{tikzpicture}
\end{figure}

\begin{example} \label{Examp2}
Let $(N,A,M) \in \mathcal{R}$ be the ranking problem in Figure \ref{Fig2} with the set of objects $N = \{ X_1,X_2,X_3 \}$ and tournament matrix
\[
T =
\begin{pmatrix}
    0     & 1.5   & 0.5 \\
    0.5   & 0     & 3 \\
    0.5   & 0     & 0 \\
\end{pmatrix}.
\]
\end{example}

Here $m=3$ and $n=3$, therefore the reasonable upper bound of $\varepsilon$ is $1/3$. Let choose it as a fixed parameter:
\[
\mathbf{x}(1/3)(N,A,M) = \left[ 2.0000;\, 2.0000;\, -4.0000 \right]^\top \text{, and }
\]
\[
\mathbf{x}(1/3)(N,2A,2M) = \left[ 4.5352;\, 3.9437;\, -8.4789  \right]^\top,
\]
implying $X_1 \sim_{(N,A,M)}^{\mathbf{x}(1/3)} X_2$ but $X_1 \succ_{(N,2A,2M)}^{\mathbf{x}(1/3)} X_2$.
\end{proof}

Now allow $\varepsilon$ to depend on the matches matrix $M$.

\begin{proposition} \label{Prop55}
The generalized row sum method satisfies $RCS$ if $\varepsilon$ is inversely proportional to the number of added ranking problems.
\end{proposition}

\begin{proof}
Let $\mathbf{x}(\varepsilon)(N,A,M) = \mathbf{x}(\varepsilon)$, $\mathbf{x}(\varepsilon)(N,A',M) = \mathbf{x}(\varepsilon)'$ and $\mathbf{x}(\varepsilon)(N,A+A',M+M) = \mathbf{x}(\varepsilon)''$.
It yields from some basic calculations:
\begin{eqnarray*}
\mathbf{x}(\varepsilon / 2)'' & = & (1 + \varepsilon m n) (I + \varepsilon L)^{-1} \mathbf{s}(N,A+A',M+M) = \\
& = & (1 + \varepsilon m n) (I + \varepsilon L)^{-1} \left[ \mathbf{s}(N,A,M) + \mathbf{s}(N,A',M) \right] = \mathbf{x}(\varepsilon) + \mathbf{x}(\varepsilon)'.
\end{eqnarray*}
\end{proof}

Proposition~\ref{Prop55} suggests that generalized row sum should be applied with a parameter somewhat inversely proportional to the number of comparisons.

\begin{remark} \label{Rem5}
Generalized row sum with $\varepsilon$ at the reasonable upper bound of $1 / \left[ m(n-2) \right]$ satisfies $RCS$.
\end{remark}

\begin{lemma} \label{Lemma59}
Fair bets and dual fair bets methods violate $RCS$.
\end{lemma}

\begin{proof}
It is a consequence of Lemmata~\ref{Lemma32} and \ref{Lemma34} together with Proposition~\ref{Prop53}: since they meet $SYM$ but violate $INV$, they cannot satisfy $RCS$.
\end{proof}

\begin{proposition} \label{Prop56}
Copeland fair bets method violates $RCS$.
\end{proposition}

\begin{proof}

\begin{figure}[htbp]
\centering
\caption{Ranking problems of Example~\ref{Examp7}}
\label{Fig7}
  
\begin{subfigure}{.5\textwidth}
  \centering
  \caption{Ranking problem $(N,T)$}
  \label{Fig7a}
\begin{tikzpicture}[scale=1, auto=center, transform shape, >=triangle 45]
\tikzstyle{every node}=[draw,shape=rectangle];
  \node (n1) at (90:3) {$X_1$};
  \node (n2) at (0:3) {$X_2$};
  \node (n3) at (180:3) {$X_3$};

\draw[->,transform canvas={xshift=1ex}](n1) -- (n2);
\draw[->](n1) -- (n2);
\draw[->,transform canvas={xshift=-1ex}](n1) -- (n2);
\draw[->](n2) -- (n3);
\draw[->,transform canvas={xshift=1.5ex}](n3) -- (n1);
\draw[->,transform canvas={xshift=0.5ex}](n3) -- (n1);
\draw[->,transform canvas={xshift=-0.5ex}](n3) -- (n1);
\draw[->,transform canvas={xshift=-1.5ex}](n3) -- (n1);
\end{tikzpicture}
\end{subfigure}
\begin{subfigure}{.5\textwidth}
  \centering
  \caption{Ranking problem $(N,T')$}
  \label{Fig7b}
\begin{tikzpicture}[scale=1, auto=center, transform shape, >=triangle 45]
\tikzstyle{every node}=[draw,shape=rectangle];
  \node (n1) at (90:3) {$X_1$};
  \node (n2) at (0:3) {$X_2$};
  \node (n3) at (180:3) {$X_3$};

\draw[transform canvas={xshift=1ex}](n1) -- (n2);
\draw[->](n2) -- (n1);
\draw[transform canvas={xshift=-1ex}](n1) -- (n2);
\draw[->](n3) -- (n2);
\draw[transform canvas={xshift=1.5ex}](n3) -- (n1);
\draw[transform canvas={xshift=0.5ex}](n3) -- (n1);
\draw[transform canvas={xshift=-0.5ex}](n3) -- (n1);
\draw[transform canvas={xshift=-1.5ex}](n3) -- (n1);
\end{tikzpicture}
\end{subfigure}
\end{figure}
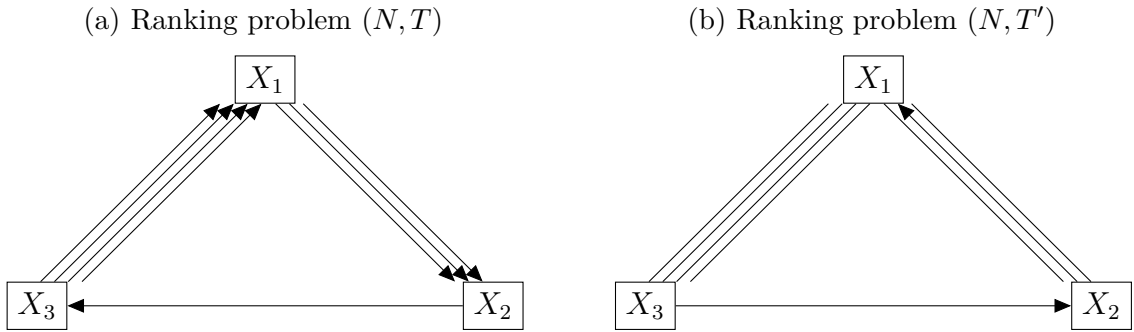

\begin{example} \label{Examp7}
Let $(N,T),(N,T') \in \mathcal{R}$ be the ranking problems in Figure~\ref{Fig7} with the set of objects $N = \{ X_1,X_2,X_3 \}$, tournament and matches matrices
\[
T =
\begin{pmatrix}
    0     & 3     & 0 \\
    0     & 0     & 1 \\
    4     & 0     & 0 \\
\end{pmatrix} ,\qquad
T' =
\begin{pmatrix}
    0     & 1     & 2 \\
    2     & 0     & 0 \\
    2     & 1     & 0 \\
\end{pmatrix} \qquad \text{and} \qquad
M = M' =
\begin{pmatrix}
    0     & 3     & 4 \\
    3     & 0     & 1 \\
    4     & 1     & 0 \\
\end{pmatrix}.
\]
Let $(N,T'') = (N,T+T') \in \mathcal{R}$ be the sum of these two ranking problems.
\end{example}

\begin{table}[htbp]
\centering
\caption{Fair bets and associated rating vectors of Example~\ref{Examp7}}
\label{Table5}
\begin{footnotesize}
    \begin{tabularx}{\textwidth}{C rrr rrr rrr} \toprule
    & $\mathbf{fb}(T)$     & $\mathbf{dfb}(T)$  & $\mathbf{Cfb}(T)$   & $\mathbf{fb}(T')$   & $\mathbf{dfb}(T')$   & $\mathbf{Cfb}(T')$ & $\mathbf{fb}(T'')$   & $\mathbf{dfb}(T'')$   & $\mathbf{Cfb}(T'')$ \\
    \midrule
    $X_1$  & $3/19$  & $-1/3$   & $-10/57$ & $2/7$   & $-6/15$  & $-12/105$ & $7/29$  & $-2/6$  & $-16/174$ \\
    $X_2$  & $4/19$  & $-1/3$   & $-7/57$  & $2/7$   & $-5/15$  & $-5/105$  & $6/29$  & $-3/6$  & $-51/174$ \\
    $X_3$  & $12/19$ & $-1/3$   & $17/57$  & $3/7$   & $-4/15$  	& $17/105$ 	& $16/29$ & $-1/6$  & $67/174$ \\
    \bottomrule
    \end{tabularx} 
\end{footnotesize}
\end{table}

The rating vectors are given in Table~\ref{Table5}: $X_1 \prec_{(N,T)}^{\mathbf{Cfb}} X_2$ and $X_1 \prec_{(N,T')}^{\mathbf{Cfb}} X_2$, but $X_1 \succ_{(N,T'')}^{\mathbf{Cfb}} X_2$.
\end{proof}

Strengthening of flatness preservation in order to get $EP$ seems to be futile. It is not surprising since equal rating of two objects may occur accidentally.
On the other side, restricting consistency by filtering out the comparison structure proves to be fruitful, at least in the case of least squares and generalized row sum with a proper parameter choice. But it is not enough to achieve positive results even for Copeland fair bets, which violates result consistency still in the most simple instance of three objects.

\subsection{The round-robin case}

Another weakening of consistency is offered by restricting its domain to a properly chosen subset of ranking problems.
Now the special case of round-robin ranking problems is analysed, when all pairs of objects have the same number of comparisons, therefore a significant difficulty of paired-comparison based ranking is eliminated. Note that the set $\mathcal{R}^R$ is closed under summation.

\begin{lemma} \label{Lemma510}
The generalized row sum and least squares methods satisfy $CS$ (therefore $EP$ and $RCS$) on the set $\mathcal{R}^R$.
\end{lemma}

\begin{proof}
Due to axioms \emph{agreement} \citep[Property~3]{Chebotarev1994} and \emph{score consistency} \citep{Gonzalez-DiazHendrickxLohmann2013}, both the generalized row sum and least squares methods coincide with the score on this set of problems, so Lemma~\ref{Lemma51} holds.
\end{proof}

Lemma~\ref{Lemma510} shows that lack of additivity in Example~\ref{Examp5} is due to the different structure of the comparison multigraphs.

\begin{lemma} \label{Lemma512}
Fair bets, dual fair bets and Copeland fair bets methods violate $EP$ even on the set $\mathcal{R}^R$.
\end{lemma}

\begin{proof}
Both $(N,T)$ and $(N,T')$ are round-robin ranking problems in Example~\ref{Examp6}.
\end{proof}

\begin{lemma} \label{Lemma514}
Fair bets and dual fair bets methods violate $RCS$ on the set $\mathcal{R}^R$.
\end{lemma}

\begin{proof}
The argument of Lemma~\ref{Lemma59} is valid because they violate $INV$ on the set $\mathcal{R}^R$ according to Lemma~\ref{Lemma34}.
\end{proof}

\begin{proposition} \label{Prop57}
Copeland fair bets methods violate $RCS$ on the set $\mathcal{R}^R$.
\end{proposition}

\begin{proof}

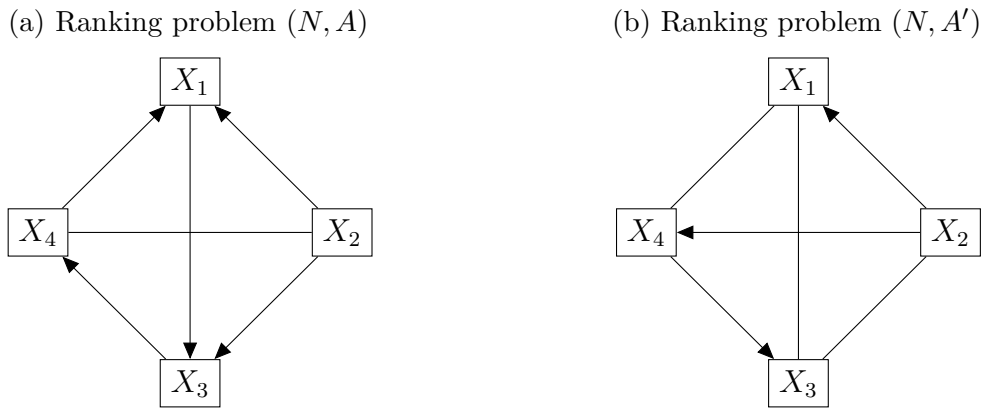
\begin{figure}[htbp]
\centering
\caption{Ranking problems of Example~\ref{Examp8}}
\label{Fig8}
  
\begin{subfigure}{.5\textwidth}
  \centering
  \caption{Ranking problem $(N,A)$}
  \label{Fig8a}
\begin{tikzpicture}[scale=1, auto=center, transform shape, >=triangle 45]
\tikzstyle{every node}=[draw,shape=rectangle];
  \node (n1) at (90:2) {$X_1$};
  \node (n2) at (0:2)  {$X_2$};
  \node (n3) at (270:2) {$X_3$};
  \node (n4) at (180:2)  {$X_4$};

\draw[->](n2) -- (n1);
\draw[->](n1) -- (n3);
\draw[->](n4) -- (n1);
\draw[->](n2) -- (n3);
\draw(n2) -- (n4);
\draw[->](n3) -- (n4);
\end{tikzpicture}
\end{subfigure}
\begin{subfigure}{.5\textwidth}
  \centering
  \caption{Ranking problem $(N,A')$}
  \label{Fig8b}
\begin{tikzpicture}[scale=1, auto=center, transform shape, >=triangle 45]
\tikzstyle{every node}=[draw,shape=rectangle];
  \node (n1) at (90:2) {$X_1$};
  \node (n2) at (0:2)  {$X_2$};
  \node (n3) at (270:2) {$X_3$};
  \node (n4) at (180:2)  {$X_4$};

\draw[->](n2) -- (n1);
\draw(n1) -- (n3);
\draw(n4) -- (n1);
\draw(n2) -- (n3);
\draw[->](n2) -- (n4);
\draw[->](n4) -- (n3);
\end{tikzpicture}
\end{subfigure}
\end{figure}

\begin{example} \label{Examp8}
Let $(N,T),(N,T') \in \mathcal{R}$ be the ranking problems in Figure~\ref{Fig8} with the set of objects $N = \{ X_1,X_2,X_3,X_4 \}$, tournament and matches matrices
\[
T =
\begin{pmatrix}
    0     & 0     & 1     & 0 \\
    1     & 0     & 1     & 0.5 \\
    0     & 0     & 0     & 1 \\
    1     & 0.5   & 0     & 0 \\
\end{pmatrix}, \quad
T' =
\begin{pmatrix}
    0     & 0     & 0.5   & 0.5 \\
    1     & 0     & 0.5   & 1 \\
    0.5   & 0.5   & 0     & 0 \\
    0.5   & 0     & 1     & 0 \\
\end{pmatrix} \quad \text{and} \quad
M = M' =
\begin{pmatrix}
    0     & 1     & 1     & 1 \\
    1     & 0     & 1     & 1 \\
    1     & 1     & 0     & 1 \\
    1     & 1     & 1     & 0 \\
\end{pmatrix}.
\]
Let $(N,T'') = (N,T+T') \in \mathcal{R}^R$ be the sum of these two ranking problems.
\end{example}

\begin{table}[htbp]
\centering
\caption{Fair bets and associated rating vectors of Example \ref{Examp8}}
\label{Table6}
\begin{footnotesize}
\noindent\makebox[\textwidth]{
    \begin{tabularx}{1\textwidth}{C rrr rrr rrr} \toprule
    & $\mathbf{fb}(T)$     & $\mathbf{dfb}(T)$  & $\mathbf{Cfb}(T)$   & $\mathbf{fb}(T')$   & $\mathbf{dfb}(T')$   & $\mathbf{Cfb}(T')$ & $\mathbf{fb}(T'')$   & $\mathbf{dfb}(T'')$   & $\mathbf{Cfb}(T'')$ \\
    \midrule
    $X_1$  & $1/17$  & $-6/19$  & $-83/323$ &    $5/64$  &   $-23/64$  & $-9/32$  & $17/236$  & $-79/244$ & $- 906/3599$ \\
    $X_2$  & $10/17$ & $-1/19$  & $173/323$ &   $39/64$  &    $-5/64$  & $17/32$  & $145/236$ & $-15/244$ & $1990/3599$ \\
    $X_3$  & $2/17$  & $-7/19$  & $-81/323$ &   $11/64$  &   $-25/64$  & $-7/32$  & $31/236$  & $-97/244$ & $-958/3599$ \\
    $X_4$  & $4/17$  & $-5/19$  & $-9/323$  &    $9/64$  &   $-11/64$  & $-1/32$  & $43/236$  & $-53/244$ & $-126/3599$ \\
    \bottomrule
    \end{tabularx} }
\end{footnotesize}
\end{table}

The rating vectors are given in Table~\ref{Table6}: $X_1 \prec_{(N,T)}^{\mathbf{Cfb}} X_3$ and $X_1 \prec_{(N,T')}^{\mathbf{Cfb}} X_3$, but $X_1 \succ_{(N,T'')}^{\mathbf{Cfb}} X_3$.
\end{proof}

\begin{lemma} \label{Lemma513}
Fair bets, dual fair bets and Copeland fair bets methods violate $CS$ on the set $\mathcal{R}^R$.
\end{lemma}

\begin{proof}
It comes from Lemma~\ref{Lemma514} and Proposition~\ref{Prop57} together with Corollary~\ref{Col4}.
\end{proof}

In the case of round-robin ranking problems, generalized row sum and least squares coincide with the score, so they have a 'perfect' performance regarding additivity.
Rankings according to fair bets, dual fair bets and Copeland fair bets may be reversed by adding two round-robin ranking problems even if there are only four objects, despite the latter satisfies inversion.

\section{Additivity and irrelevant comparisons} \label{Sec4}

From the viewpoint of additivity, score method seems to be flawless. However, consistency may have some unintended consequences, which are difficult to accept. This section deals with the connection of additivity with other axioms.

\subsection{Independence of irrelevant matches and results} \label{Sec41}

\begin{definition} \label{Def61}
\emph{Independence of irrelevant matches} ($IIM$):
Let $(N,T) \in \mathcal{R}$ be a ranking problem and $X_i, X_j, X_k, X_\ell \in N$ be four different objects.
Let $f: \mathcal{R} \to \mathbb{R}^n$ be a scoring procedure such that $f_i(N,T) \geq f_j(N,T)$ and $(N,T') \in \mathcal{R}$ be a ranking problem identical to $(N,T)$ except for $t'_{k \ell} \neq t_{k \ell}$.
$f$ is called \emph{independent of irrelevant matches} if $f_i(N,T') \geq f_j(N,T')$.
\end{definition}

\begin{remark} \label{Rem6}
Property $IIM$ has a meaning if $n \geq 4$.
\end{remark}

Sequential application of independence of irrelevant matches can lead to any ranking problem $(N,T') \in \mathcal{R}$, for which $t_{gh}' = t_{gh}$ if $\{ X_g,X_h \} \cap \{ X_i, X_j \} \neq \emptyset$, but all other paired comparisons are arbitrary.
$IIM$ means that all comparisons not involving the two objects chosen are irrelevant from the perspective of their relative ranking.

This property appears as \emph{independence} in \citet[Axiom~III]{Rubinstein1980} and \citet[Axiom~5]{NitzanRubinstein1981} in the case of round-robin ranking problems.
The name independence of irrelevant matches was introduced by \citet{Gonzalez-DiazHendrickxLohmann2013}., 
\citet[Definition~8.4]{AltmanTennenholtz2008} use a stronger axiom called \emph{Arrow's independence of irrelevant alternatives} by permitting modifications of comparisons involving $X_i$ and $X_j$ if $t_{ih} - t_{ih}' = t_{jh} - t_{jh}'$ holds for all $X_h \in N \setminus \{ X_i, X_j \}$.

Decomposition of the tournament matrix into the results matrix $A$ and matches matrix $M$ makes possible to weaken $IIM$.

\begin{definition} \label{Def62}
\emph{Independence of irrelevant results} ($IIR$):
Let $(N,A,M) \in \mathcal{R}$ be a ranking problem and $X_i, X_j, X_k, X_\ell \in N$ be four different objects.
Let $f: \mathcal{R} \to \mathbb{R}^n$ be a scoring procedure such that $f_i(N,A,M) \geq f_j(N,A,M)$ and $(N,A',M) \in \mathcal{R}$ be a ranking problem identical to $(N,A,M)$ except for the result $a'_{k \ell} \neq a_{k \ell}$.
$f$ is called \emph{independent of irrelevant results} if $f_i(N,A',M) \geq f_j(N,A',M)$.
\end{definition}

Sequential application of independence of irrelevant matches can result in any ranking problem $(N,A',M) \in \mathcal{R}$, for which $a_{gh}' = a_{gh}$ if $\{ X_g,X_h \} \cap \{ X_i, X_j \} \neq \emptyset$, but all other paired comparisons are arbitrary.
However, this axiom does not allow for a change in the number of matches between two objects (in the case of $IIM$, $t'_{k \ell} \neq t_{k \ell}$ means that $a'_{k \ell} \neq a_{k \ell}$ and $m'_{k \ell} \neq m_{k \ell}$ may occur).

Note also that $IIR$ does not affect the connectedness of the ranking problem, however, it may influence irreducibility.

\begin{corollary} \label{Col7}
$IIM$ implies $IIR$.
\end{corollary}

\begin{proof}
It follows from Definitions~\ref{Def61} and \ref{Def62}.
\end{proof}

\begin{remark} \label{Rem}
$IIM$ and $IIR$ coincide on the set of round-robin ranking problems $\mathcal{R}^R$.
\end{remark}

\begin{lemma} \label{Lemma61}
The score method satisfies $IIM$.
\end{lemma}

\begin{proof}
It follows from Definition~\ref{Def21}.
\end{proof}

\begin{proposition} \label{Prop61}
The generalized row sum, least squares, fair bets, dual fair bets and Copeland fair bets methods violate $IIR$.
\end{proposition}

\begin{proof}

\begin{figure}[htbp]
\centering
\caption{Ranking problems of Example~\ref{Examp9}}
\label{Fig9}
  
\begin{subfigure}{.5\textwidth}
  \centering
  \subcaption{Ranking problem $(N,A,M)$}
  \label{Fig9a}
\begin{tikzpicture}[scale=1, auto=center, transform shape, >=triangle 45]
\tikzstyle{every node}=[draw,shape=rectangle]; 
  \node (n1) at (90:2) {$X_1$};
  \node (n2) at (0:2)  {$X_2$};
  \node (n3) at (270:2) {$X_3$};
  \node (n4) at (180:2)  {$X_4$};

  \foreach \from/\to in {n1/n2,n1/n4,n2/n3}
    \draw (\from) -- (\to);
  \draw [->] (n4) -- (n3);
\end{tikzpicture}
\end{subfigure}
\begin{subfigure}{.5\textwidth}
  \centering
  \subcaption{Ranking problem $(N,A',M)$}
  \label{Fig9b}
\begin{tikzpicture}[scale=1, auto=center, transform shape, >=triangle 45]
\tikzstyle{every node}=[draw,shape=rectangle];
  \node (n1) at (90:2) {$X_1$};
  \node (n2) at (0:2)  {$X_2$};
  \node (n3) at (270:2) {$X_3$};
  \node (n4) at (180:2)  {$X_4$};

  \foreach \from/\to in {n1/n2,n1/n4,n2/n3}
    \draw (\from) -- (\to);
  \draw [->] (n3) -- (n4);
\end{tikzpicture}
\end{subfigure}
\end{figure}

\begin{example} \label{Examp9}
Let $(N,A,M),(N,A',M) \in \mathcal{R}$ be the ranking problems in Figure~\ref{Fig9} with set of objects $N = \{ X_1,X_2,X_3,X_4 \}$, tournament and matches matrices
\[
T =
\begin{pmatrix}
    0     & 0.5   & 0     & 0.5 \\
    0.5   & 0     & 0.5   & 0   \\
    0     & 0.5   & 0     & 0   \\
    0.5   & 0     & 1     & 0   \\
\end{pmatrix},\, 
T' =
\begin{pmatrix}
    0     & 0.5   & 0     & 0.5 \\
    0.5   & 0     & 0.5   & 0   \\
    0     & 0.5   & 0     & 1   \\
    0.5   & 0     & 0     & 0   \\
\end{pmatrix} \text{ and }
M = M' =
\begin{pmatrix}
    0     & 1     & 0     & 1 \\
    1     & 0     & 1     & 0   \\
    0     & 1     & 0     & 1   \\
    1     & 0     & 0     & 0   \\
\end{pmatrix},
\]
where $a_{34}' \neq a_{34}$ but $m_{34}' = m_{34}$.
\end{example}

$IIM$ requires that $f_1(N,A,M) \geq f_2(N,A,M) \Leftrightarrow f_1(N,A',M) \geq f_2(N,A',M)$. Let $\mathbf{x}(\varepsilon)(N,A,M) = \mathbf{x}(\varepsilon)$, $\mathbf{x}(\varepsilon)(N,A',M') = \mathbf{x}(\varepsilon)'$ and $\mathbf{q}(N,A,M) = \mathbf{q}$, $\mathbf{q}(N,A',M') = \mathbf{q}'$.
Here $m=1$ and $n=4$, therefore
\[
x_1(\varepsilon) = x_2(\varepsilon)' = (1 + \varepsilon m n) \frac{\varepsilon}{(1 + 2\varepsilon)(1 + 4\varepsilon)} = \frac{\varepsilon}{1 + 2\varepsilon} \text{ and}
\]
\[
x_1(\varepsilon)' = x_2(\varepsilon) = (1 + \varepsilon m n) \frac{-\varepsilon}{(1 + 2\varepsilon)(1 + 4\varepsilon)} = \frac{-\varepsilon}{1 + 2\varepsilon},
\]
that is, $X_1 \succ_{(N,A,M)}^{\mathbf{x}(\varepsilon)} X_2$ but $X_1 \prec_{(N,A',M)}^{\mathbf{x}(\varepsilon)} X_2$.

For the least squares method, on the basis of Lemma~\ref{Lemma22}:
\[
q_1 = \frac{\lim_{\varepsilon \to \infty} x_1(\varepsilon)}{m n} = q_2' = \frac{\lim_{\varepsilon \to \infty} x_2(\varepsilon)'}{m n} = \frac{1}{2} \cdot \frac{1}{1 \cdot 4} = \frac{1}{8} \text{ and}
\]
\[
q_1' = \frac{\lim_{\varepsilon \to \infty} x_1(\varepsilon)'}{m n} = q_2 = \frac{\lim_{\varepsilon \to \infty} x_2(\varepsilon)}{m n} = -\frac{1}{2} \cdot \frac{1}{1 \cdot 4} = -\frac{1}{8}.
\]
Hence $X_1 \succ_{(N,A,M)}^{\mathbf{q}} X_2$ but $X_1 \prec_{(N,A',M)}^{\mathbf{q}} X_2$.

\begin{table}[htbp]
\centering
\caption{Fair bets and associated rating vectors of Example~\ref{Examp9}}
\label{Table7}
\begin{footnotesize}
    \begin{tabularx}{0.8\textwidth}{C rrr rrr} \toprule
    & $\mathbf{fb}(N,T)$     & $\mathbf{dfb}(N,T)$  & $\mathbf{Cfb}(N,T)$   & $\mathbf{fb}(N,T')$   & $\mathbf{dfb}(N,T')$   & $\mathbf{Cfb}(N,T')$ \\
    \midrule
    $X_1$   & $5/16$  & $-3/16$  & $1/8$   & $3/16$  & $-5/16$  & $-1/8$  \\
    $X_2$   & $3/16$  & $-5/16$  & $-1/8$  & $5/16$  & $-3/16$  & $1/8$   \\
    $X_3$   & $1/16$  & $-7/16$  & $-3/8$  & $7/16$  & $-1/16$  & $3/8$   \\
    $X_4$   & $7/16$  & $-1/16$  & $3/8$   & $1/16$  & $-7/16$  & $-3/8$  \\
    \bottomrule
    \end{tabularx} 
\end{footnotesize}
\end{table}

The other three rating vectors are given in Table~\ref{Table7}: $X_1 \succ_{(N,T)} X_2$ and $X_1 \prec_{(N,T')} X_2$ for the three methods.
\end{proof}

\begin{remark}
The two ranking problems of Example~\ref{Examp9} coincide with the permutation $\sigma(X_1) = X_2$ and $\sigma(X_3) = X_4$. Then independence of irrelevant matches demands that $f_1(N,A,M) = f_2(N,A,M)$, violated by all ranking methods discussed except for the score.
\end{remark}

\begin{lemma} \label{Lemma62}
The generalized row sum, least squares, fair bets, dual fair bets and Copeland fair bets methods violate $IIM$.
\end{lemma}

\begin{proof}
It comes from Proposition~\ref{Prop61} and Corollary~\ref{Col7}.
\end{proof}

\begin{lemma} \label{Lemma63}
The generalized row sum and least squares methods satisfy $IIM$ on the set $\mathcal{R}^R$.
\end{lemma}

\begin{proof}
Due to axioms \emph{agreement} \citep[Property~3]{Chebotarev1994} and \emph{score consistency} \citep{Gonzalez-DiazHendrickxLohmann2013}, both the generalized row sum and least squares methods coincide with the score on this set of problems, so Lemma~\ref{Lemma61} holds.
\end{proof}

\begin{proposition} \label{Prop62}
Fair bets, dual fair bets and Copeland fair bets methods violate $IIR$ ($IIM$) even on the set $\mathcal{R}^R$.
\end{proposition}

\begin{proof}

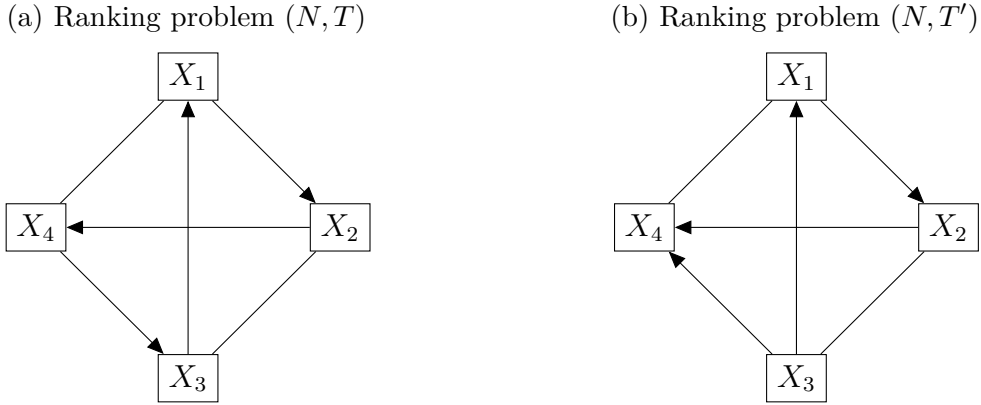
\begin{figure}[htbp]
\centering
\caption{Ranking problems of Example~\ref{Examp10}}
\label{Fig10}
  
\begin{subfigure}{.5\textwidth}
  \centering
  \caption{Ranking problem $(N,T)$}
  \label{Fig10a}
\begin{tikzpicture}[scale=1, auto=center, transform shape, >=triangle 45]
\tikzstyle{every node}=[draw,shape=rectangle];
  \node (n1) at (90:2) {$X_1$};
  \node (n2) at (0:2)  {$X_2$};
  \node (n3) at (270:2) {$X_3$};
  \node (n4) at (180:2)  {$X_4$};

\draw[->](n1) -- (n2);
\draw[->](n3) -- (n1);
\draw(n1) -- (n4);
\draw(n2) -- (n3);
\draw[->](n2) -- (n4);
\draw[->](n4) -- (n3);
\end{tikzpicture}
\end{subfigure}
\begin{subfigure}{.5\textwidth}
  \centering
  \caption{Ranking problem $(N,T')$}
  \label{Fig10b}
\begin{tikzpicture}[scale=1, auto=center, transform shape, >=triangle 45]
\tikzstyle{every node}=[draw,shape=rectangle];
  \node (n1) at (90:2) {$X_1$};
  \node (n2) at (0:2)  {$X_2$};
  \node (n3) at (270:2) {$X_3$};
  \node (n4) at (180:2)  {$X_4$};

\draw[->](n1) -- (n2);
\draw[->](n3) -- (n1);
\draw(n1) -- (n4);
\draw(n2) -- (n3);
\draw[->](n2) -- (n4);
\draw[->](n3) -- (n4);
\end{tikzpicture}
\end{subfigure}
\end{figure}

\begin{example} \label{Examp10}
Let $(N,T),(N,T') \in \mathcal{R}^R$ be the round-robin ranking problems in Figure~\ref{Fig10} with the set of objects $N = \{ X_1,X_2,X_3,X_4 \}$, tournament and matches matrices
\[
T =
\begin{pmatrix}
    0     & 1     & 0     & 0.5 \\
    0     & 0     & 0.5   & 1 \\
    1     & 0.5   & 0     & 0 \\
    0.5   & 0     & 1     & 0 \\
\end{pmatrix}, \,
T' =
\begin{pmatrix}
    0     & 1     & 0     & 0.5 \\
    0     & 0     & 0.5   & 1 \\
    1     & 0.5   & 0     & 1 \\
    0.5   & 0     & 0     & 0 \\
\end{pmatrix} \text{ and }
M = M' =
\begin{pmatrix}
    0     & 1     & 1     & 1 \\
    1     & 0     & 1     & 1   \\
    1     & 1     & 1     & 1   \\
    1     & 1     & 1     & 0   \\
\end{pmatrix},
\]
where $a_{34}' \neq a_{34}$ but $m_{34}' \neq m_{34}$.
\end{example}

\begin{table}[htbp]
\centering
\caption{Fair bets and associated rating vectors of Example~\ref{Examp10}}
\label{Table8}
\begin{footnotesize}
    \begin{tabularx}{0.8\textwidth}{C rrr rrr} \toprule
    & $\mathbf{fb}(N,T)$     & $\mathbf{dfb}(N,T)$  & $\mathbf{Cfb}(N,T)$   & $\mathbf{fb}(N,T')$   & $\mathbf{dfb}(N,T')$   & $\mathbf{Cfb}(N,T')$ \\
    \midrule
    $X_1$   & $1/4$  & $-1/4$  & $0$  & $5/32$  & $-7/32$  & $-1/16$  \\
    $X_2$   & $1/4$  & $-1/4$  & $0$  & $7/32$  & $-5/32$  & $1/16$   \\
    $X_3$   & $1/4$  & $-1/4$  & $0$  & $19/32$ & $-1/32$  & $9/16$   \\
    $X_4$   & $1/4$  & $-1/4$  & $0$  & $1/32$  & $-19/32$ & $-9/16$  \\
    \bottomrule
    \end{tabularx} 
\end{footnotesize}
\end{table}

$IIR$ requires that $f_1(N,A,M) \geq f_2(N,A,M) \Leftrightarrow f_1(N,A',M) \geq f_2(N,A',M)$.
The rating vectors are given in Table~\ref{Table8}: $X_1 \succeq_{(N,T)} X_2$ and $X_1 \prec_{(N,T')} X_2$ for the three methods.
\end{proof}

Hence, similarly to consistency, generalized row sum and least squares satisfy $IIR$ on the set of round-robin ranking problems, while fair bets, dual fair bets and Copeland fair bets break it even on this restricted domain.

\subsection{Connection to additivity} \label{Sec42}

Take a look at Example~\ref{Examp9} (Figure~\ref{Fig9}). It seems strange to require that objects $X_1$ and $X_2$ have the same rank in both problems, which is an implication of $IIM$.
Therefore, \citet[p.~165]{Gonzalez-DiazHendrickxLohmann2013} consider independence of irrelevant matches to be a drawback of the score method because outside the subdomain of round-robin ranking problems, it makes sense if the scoring procedure is responsive to the strength of the opponents.
However, it turns out that $IIM$ is closely linked to additivity.

\begin{theorem} \label{Theo61}
$NEU$, $SYM$ and $CS$ imply $IIM$.
\end{theorem}

\begin{proof}
For the round-robin case, see \citet[Lemma~3]{NitzanRubinstein1981}.

Assume to the contrary, and let $(N,A,M) \in \mathcal{R}$ be a ranking problem, $X_i, X_j, X_k, X_\ell \in N$ be four different objects such that $f_i(N,A,M) \geq f_j(N,A,M)$, and $(N,A',M') \in \mathcal{R}$ is identical to $(N,A,M)$ except for the result $a'_{k \ell} \neq a_{k \ell}$ and match $m'_{k \ell} \neq m_{k \ell}$ such that $f_i(N,A',M') < f_j(N,A',M')$.

Corollary~\ref{Col6} implies that a symmetric and consistent scoring procedure satisfies $INV$, hence $f_i(N,-A,M) \leq f_j(N,-A,M)$.
Denote by $\sigma: N \rightarrow N$ the permutation $\sigma(X_i) = X_j$, $\sigma(X_j) = X_i$, and $\sigma(X_k) = X_k$ for all $X_k \in N \setminus \{ X_i,X_j \}$. Neutrality implies $f_i \left[ \sigma(N,A,M) \right] \leq f_j \left[ \sigma(N,A,M) \right]$, and $f_i \left[ \sigma(N,-A',M') \right] < f_j \left[ \sigma(N,-A',M') \right]$ due to inversion and Remark~\ref{Rem3}. With the definitions $A'' = \sigma(A) - \sigma(A') - A + A' = O$ and $M'' = \sigma(M) + \sigma(M') + M + M'$,
\[
(N,A'',M'') = \sigma(N,A,M) + \sigma(N,-A',M') - (N,A,M) + (N,A',M').
\]
Symmetry implies $f_i(N,A'',M'') = f_j(N,A'',M'')$ since $A'' = 0$, but $f_i(N,A'',M'') < f_j(N,A'',M'')$ from consistency, which is a contradiction. 
\end{proof}

\begin{remark}
$NEU$, $SYM$ and $RCS$ \emph{do not} imply $IIR$ despite that the proof of Theorem~\ref{Theo61} can almost be followed. According to Proposition~\ref{Prop53}, a symmetric and result consistent scoring procedure also satisfies $INV$, but result consistency \emph{does not} provide $f_i(N,A'',M'') < f_j(N,A'',M'')$ even if $M = M'$ (guaranteed in the case of $IIR$) due to $M \neq \sigma(M)$ in general.

Note that $M = \sigma(M)$ is equivalent to $m_{ik} = m_{jk}$ for all $X_k \in N \setminus \{ X_i, X_j \}$. Then $NEU$, $SYM$ and $RCS$ still imply $IIR$, so generalized row sum and least squares should satisfy independence of irrelevant results with respect to such objects $X_i$ and $X_j$. In fact, according to the property \emph{homogeneous treatment of victories} \citep{Gonzalez-DiazHendrickxLohmann2013}, in this case they result in $X_i \succeq X_j$ if and only if $s_i(N,A,M) \geq s_j(N,A,M)$: when two objects have the same number of comparisons against all the other objects, they are ranked according to their scores.\footnote{~Formally, \citet{Gonzalez-DiazHendrickxLohmann2013} prove homogeneous treatment of victories only for generalized row sum with $\varepsilon = 1 / \left[ m (n-2) \right]$, but it remains valid for any $\varepsilon > 0$.}
As $m_{ik} = m_{jk}$ for all $X_k \in N \setminus \{ X_i, X_j \}$ holds for any $X_i, X_j \in N$ in round-robin ranking problems, it highlights that generalized row sum and least squares satisfy $IIM$ on the domain of $\mathcal{R}^R$.
\end{remark}

Axioms $NEU$ and $SYM$ are difficult to debate, therefore Theorem~\ref{Theo61} implies $CS$ is a property one would rather not have in the general case. It reinforces the significance of Section~\ref{Sec3}; weakening of consistency seems to be desirable in order to avoid independence of irrelevant matches (results). 

\section{Conclusions} \label{Sec5}

\begin{table}[htbp]
\centering
\caption{Axiomatic properties of ranking methods}
\label{Table9}
\begin{ThreePartTable}
\noindent\makebox[\textwidth]{
	\begin{tabularx}{1\textwidth}{l CCCCCC} \toprule
    Property & Score$^{\dag}$ & Generalized row sum$^{\ddag}$ & Least squares & Fair bets / dual fair bets$^\ast$ & Copeland fair bets \\ \midrule
    ($NEU$) & (\textcolor{PineGreen}{\ding{52}}) & (\textcolor{PineGreen}{\ding{52}}) & (\textcolor{PineGreen}{\ding{52}}) & (\textcolor{PineGreen}{\ding{52}}) & \textcolor{PineGreen}{\ding{52}} \\
    ($SYM$) & (\textcolor{PineGreen}{\ding{52}}) & (\textcolor{PineGreen}{\ding{52}}) & (\textcolor{PineGreen}{\ding{52}}) & (\textcolor{PineGreen}{\ding{52}}) & \textcolor{PineGreen}{\ding{52}} \\
    ($INV$) & (\textcolor{PineGreen}{\ding{52}}) & (\textcolor{PineGreen}{\ding{52}}) & (\textcolor{PineGreen}{\ding{52}}) & (\textcolor{BrickRed}{\ding{55}}) & \textcolor{PineGreen}{\ding{52}} \\ \midrule
    ($CS$) & (\textcolor{PineGreen}{\ding{52}}) & (\textcolor{BrickRed}{\ding{55}}) & (\textcolor{BrickRed}{\ding{55}}) & (\textcolor{BrickRed}{\ding{55}}) & \textcolor{BrickRed}{\ding{55}} \\
    ($FP$) & (\textcolor{PineGreen}{\ding{52}}) & (\textcolor{PineGreen}{\ding{52}}) & (\textcolor{PineGreen}{\ding{52}}) & (\textcolor{PineGreen}{\ding{52}}) & \textcolor{PineGreen}{\ding{52}} \\
    $EP$ & \textcolor{PineGreen}{\ding{52}} & \textcolor{BrickRed}{\ding{55}} & \textcolor{BrickRed}{\ding{55}} & \textcolor{BrickRed}{\ding{55}} & \textcolor{BrickRed}{\ding{55}} \\
    $RCS$ & \textcolor{PineGreen}{\ding{52}} & \textcolor{PineGreen}{\ding{52}}\textcolor{BrickRed}{\ding{55}}$^\diamond$ & \textcolor{PineGreen}{\ding{52}} & \textcolor{BrickRed}{\ding{55}} & \textcolor{BrickRed}{\ding{55}} \\ \midrule
    ($IIM$) & (\textcolor{PineGreen}{\ding{52}}) & (\textcolor{BrickRed}{\ding{55}}) & (\textcolor{BrickRed}{\ding{55}}) & (\textcolor{BrickRed}{\ding{55}}) & \textcolor{BrickRed}{\ding{55}} \\
    $IIR$ & \textcolor{PineGreen}{\ding{52}} & \textcolor{BrickRed}{\ding{55}} & \textcolor{BrickRed}{\ding{55}} & \textcolor{BrickRed}{\ding{55}} & \textcolor{BrickRed}{\ding{55}} \\ \bottomrule
    \end{tabularx} }
    \vspace{0.15cm}
    \begin{tablenotes}
        \footnotesize
        \item Axioms introduced in the literature and known results are in parenthesis (see the text for references); others are our contribution
        \item[$^{\dag}$]\citet{Gonzalez-DiazHendrickxLohmann2013} define the score method differently; their findings are in parenthesis
        \item[$^{\ddag}$]\citet{Gonzalez-DiazHendrickxLohmann2013} discuss generalized row sum only for $\varepsilon = \left[ 1 / m(n-2) \right]$; their findings are in parenthesis
        \item[$^\ast$]\citet{Gonzalez-DiazHendrickxLohmann2013} do not analyse dual fair bets; their findings are in parenthesis
        \item[$^\diamond$] Depends on the choice of $\varepsilon$; the answer is positive if the parameter is inversely proportional to the number of added ranking problems
    \end{tablenotes}
\end{ThreePartTable}
\end{table}

Our results concerning the connection of the axioms and ranking methods are summarized in Table~\ref{Table9}. Score satisfies all properties, however, $IIM$ is not favourable in the presence of missing and multiple comparisons. The findings recommend to use generalized row sum with a parameter somewhat proportional to the number of matches, for example, the upper bound of reasonable choice $1 / \left[ m(n-2) \right]$. It is not surprising given the statistical background of the method \citep{Chebotarev1994}. Then generalized row sum and least squares cannot be distinguished with respect to the properties examined.\footnote{~Some of their differences are highlighted by \citet{Gonzalez-DiazHendrickxLohmann2013}.}

A drawback of fair bets (and its dual) was eliminated by the introduction of Copeland fair bets, but it does not affect other axioms. \citet{ChebotarevShamis1999}'s analysis of \emph{self-consistent monotonicity} confirm that 'manipulation' with win-loss combining scoring procedures is not able to correct some inherent features of this class.

It has been investigated whether the ranking methods meet the properties on the restricted domain of round-robin tournaments. Since generalized row sum and least squares coincide with the score on this set, they perform perfectly -- in this case it is difficult to debate $IIM$ ($IIR$). However, the behaviour of fair bets and its peers remain unchanged even on this narrow subset, so a rank reversal may occur after adding two simple round-robin ranking problems. It seems to be a strong argument against their application.\footnote{~\citet{Gonzalez-DiazHendrickxLohmann2013} does not mention it as a drawback.}

We have also aspired to give simple counterexamples, minimal with respect to the number of objects and matches. It shows that the violation of these properties remains an issue still in the case of relatively small problems.

\begin{figure}[htbp]
\centering
\caption[Connections among the axioms]{Connections among the axioms

\vspace{0.25cm}
\footnotesize{Arrows sign implication. In certain cases some axioms together determine another such as $NEU + SYM + CS \Rightarrow IIM$. Nodes with dashed, red lines are properties introduced by us; continuous, blue lines are our results; dashed, green lines are trivial relationships.}
}
\label{Fig11}
\begin{tikzpicture}[scale=1,auto=center, transform shape, >=triangle 45]
\tikzstyle{every node}=[draw,shape=circle];
  \node (NEU) at (6,2) {$NEU$};
  \node (CS) at (0,0) {$\quad CS \quad$};
  \node (FP) at (-3,4) {$FP$};
  \node [red,dashed,thick] (EP) at (-3,0) {$EP$};
  \node [red,dashed,thick] (RCS) at (0,4) {$RCS$};
  \node (SYM) at (1,2) {$SYM$};
  \node (INV) at (4,2) {$INV$};
  \node (IIM) at (9,0) {$IIM$};  
  \node [red,dashed,thick] (IIR) at (9,4) {$IIR$};
          
  \foreach \from/\to in {CS/EP,CS/RCS,IIM/IIR,INV/SYM,EP/FP}
    \draw [->,green,dashed] (\from) -- (\to);
    
\draw[blue](RCS) -- (2,3);
\draw[blue](SYM) -- (2,3);
\draw[->,blue,thick](2,3) -- (INV);
\draw[blue](NEU) -- (7,0);
\draw[blue](CS) -- (7,0);
\draw[blue](SYM) -- (7,0);
\draw[->,blue,ultra thick](7,0) -- (IIM);
\end{tikzpicture}
\end{figure}

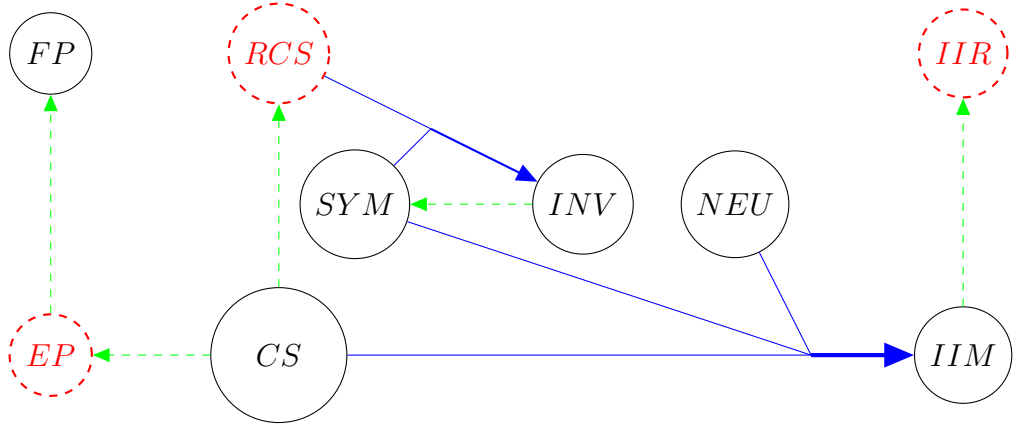
Figure~\ref{Fig11} gives a comprehensive picture about the axioms investigated. Three novel properties were introduced. $EP$ is between two extreme additivity requirements, the severe $CS$ and the weak $FP$. However, our methods show the same behaviour against equality preservation as against consistency. The other direction of mitigating $CS$, result consistency ($RCS$) -- made possible by the differentiation of results and matches matrices -- yields more success. The new setting is also responsible for the introductione of independence of irrelevant results, a weak form of independence of irrelevant matches already defined.
Relationships among the axioms shed light on some discoveries of Table~\ref{Table9}: the strong connection of $IIM$ and $CS$ justifies the violation of both properties, the violation of $INV$ by fair bets implies that it does not satisfy $RCS$.

At least two main directions of future research emerge. The first is to extend the scope of the analysis with other scoring procedures. For example, \citet{SlikkerBormvandenBrink2012} define a general framework for ranking the nodes of directed graphs, resulting in fair bets as a limit. Positional power \citep{HeringsvanderLaanTalman2005} is also worth to analyse since it is similar to least squares from a graph-theoretic point of view \citep{Csato2015a}.
The second course is to get some characterization results, an intended end goal of any axiomatic analysis.


\end{document}